\documentclass[
a4paper,reqno]{amsart}
\usepackage[T1]{fontenc}
\usepackage[english]{babel}
\usepackage{amssymb,bbm,enumerate}
\usepackage{color}
\usepackage[linktocpage=true,colorlinks=true, linkcolor=blue, citecolor=red, urlcolor=green]{hyperref}

\usepackage{float}
\restylefloat{table}

\usepackage{tikz}

\renewcommand{\phi}{\varphi}

\newcommand{\mc}[1]{\mathcal{#1}}
\newcommand{\mf}[1]{\mathfrak{#1}}
\newcommand{\mb}[1]{\mathbb{#1}}

\newcommand{\id}{\mathbbm{1}}

\newcommand{\tint}{{\textstyle\int}}

\newcommand{\ldb}{\{\!\!\{}
\newcommand{\rdb}{\}\!\!\}}

\DeclareMathOperator{\Mat}{Mat}

\DeclareMathOperator{\diag}{diag}
\DeclareMathOperator{\tr}{tr}
\DeclareMathOperator{\res}{Res}

\DeclareMathOperator{\im}{Im}

\DeclareMathOperator{\Res}{Res}

\DeclareMathOperator{\mult}{m}

\theoremstyle{plain}
\newtheorem{theorem}{Theorem}[section]
\newtheorem{lemma}[theorem]{Lemma}
\newtheorem{proposition}[theorem]{Proposition}
\newtheorem{corollary}[theorem]{Corollary}

\theoremstyle{definition}
\newtheorem{definition}[theorem]{Definition}
\newtheorem{example}[theorem]{Example}

\theoremstyle{remark}
\newtheorem{remark}[theorem]{Remark}

\numberwithin{equation}{section}

\definecolor{light}{gray}{.9}

\begin{document}

\title[A new scheme of integrability]{A new scheme of integrability for (bi)Hamiltonian PDE}

\author{Alberto De Sole}
\address{Dipartimento di Matematica, Sapienza Universit\`a di Roma,
P.le Aldo Moro 2, 00185 Rome, Italy}
\email{desole@mat.uniroma1.it}
\urladdr{www1.mat.uniroma1.it/\$$\sim$\$desole}

\author{Victor G. Kac}
\address{Dept of Mathematics, MIT,
77 Massachusetts Avenue, Cambridge, MA 02139, USA}
\email{kac@math.mit.edu}

\author{Daniele Valeri}
\address{Yau Mathematical Sciences Center, Tsinghua University, 100084 Beijing, China}
\email{daniele@math.tsinghua.edu.cn}



\begin{abstract}
We develop a new method for constructing integrable Hamiltonian hierarchies
of Lax type equations, which combines the fractional powers technique of Gelfand and Dickey,
and the classical Hamiltonian reduction technique of Drinfeld and Sokolov.
The method is based on the notion of an Adler type matrix pseudodifferential operator
and the notion of a generalized quasideterminant.
We also introduce the notion of a dispersionless Adler type series,
which is applied to the study of dispersionless Hamiltonian equations.
Non-commutative Hamiltonian equations are discussed in this framework as well.
\end{abstract}


\maketitle

\tableofcontents

\section{Introduction}\label{sec:1}

In the present paper we propose a general method of constructing
integrable (bi)Hamiltonian hierarchies of PDE's,
which combines two most famous approaches.
The first one is the Gelfand-Dickey technique,
based on the Lax pair method,
which consists in taking fractional powers of pseudodifferential operators \cite{GD76,Dic03},
and the second one is the classical Hamiltonian reduction technique, 
combined with the Zakharov-Shabat method,
as developed by Drinfeld and Sokolov \cite{DS85}.

The central notion of our method is that of a
matrix pseudodifferential operator of \emph{Adler type},
as defined in our paper \cite{DSKV15a}.
It has been derived there starting from Adler's formula
for the second Poisson structure for the $M$-th KdV hierarchy \cite{Adl79}.
A pseudodifferential operator $A(\partial)$ over a differential algebra $(\mc V,\partial)$ 
is called of \emph{Adler type},
with respect to a $\lambda$-bracket $\{\cdot\,_\lambda\,\cdot\}$ on $\mc V$, if
\begin{equation}\label{eq:adler-scalar}
\begin{split}
\{A(z)_\lambda A(w)\}
& = A(w+\lambda+\partial)\iota_z(z\!-\!w\!-\!\lambda\!-\!\partial)^{-1}A^*(\lambda-z)
\\
& - A(z)\iota_z(z\!-\!w\!-\!\lambda\!-\!\partial)^{-1}A(w)
\,.
\end{split}
\end{equation}
Recall that a $\lambda$-bracket is a bilinear (over the base field $\mb F$)
map $\{\cdot\,_\lambda\,\cdot\}:\,\mc V\times\mc V\to\mc V[\lambda]$
satisfying the sesquilinearity and Leibniz rules axioms
(see Section \ref{sec:2.0}),
$A(z)$ is the symbol of $A(\partial)$,
$A^*(\partial)$ is the formal adjoint of $A(\partial)$,
and $\iota_z$ denotes the expansion in the geometric series for large $z$.

The definition of a rectangular matrix pseudodifferential matrix operator of Adler type
is similar, see \cite{DSKV15a} or Definition \ref{def:adler} of the present paper.

The first basic property of an Adler type matrix pseudodifferential operator $A(\partial)$,
proved in \cite{DSKV15a}, is that the $\lambda$-bracket on $\mc V$
restricted to the subalgebra $\mc V_1$
generated by the coefficients of the entries of $A(\partial)$
satisfies the skewcommutativity and Jacobi identity axioms of a Poisson vertex algebra (PVA)
(in Section \ref{sec:2.0} we recall its definition).
Thus an Adler type operator automatically provides $\mc V_1$
with a Poisson structure.

For example, we note in \cite{DSKV15a}
that, for each positive integer $M$, the ``generic''
pseudodifferential operator
(resp. ``generic'' differential operator)
\begin{equation}\label{eq:intro1}
L_M(\partial)
=
\sum_{j=-\infty}^M u_j\partial^j
\qquad\Big(\text{resp. }\,\, 
L_{(M)}(\partial)
=
\sum_{j=0}^M u_j\partial^j
\,\Big)
\,,\qquad\text{ with }
u_M=1
\,,
\end{equation}
is of Adler type.
Consequently we automatically get on the algebra of differential polynomials
in $\{u_j\,|\,-\infty<j<M\}$
(resp. $\{u_j\,|\,0\leq j<M\}$)
the so called $2$-nd Poisson structure for the KP hierarchy
(resp. for the $M$-th KdV hierarchy).
This Poisson structure for the $M$-th KdV hierarchy has been conjectured by Adler
in \cite{Adl79},
and subsequently proved by a lengthy calculation in \cite{GD78}.
The second Poisson structures for the KP hierarchy
were discovered by Radul in \cite{Rad87}.
Note that these structures are different for different $M$,
though the corresponding hierarchies of equations are isomorphic.
Note also that the $1$-st Poisson structure for the KP hierarchy
was previously discovered in \cite{Wat83},
and it is well known to be easily derived from the $2$-nd structure.
In fact, this is the case in general:
if $A(\partial)+\epsilon$ is of Adler type for any constant $\epsilon$,
then we call $A(\partial)$ of bi-Adler type,
and we have on $\mc V_1$ a pencil of Poisson structures,
see Section \ref{sec:6.1}.

The second basic property of an Adler type $M\times M$-matrix pseudodifferential
operator $A(\partial)$
is that it provides a hierarchy of compatible
Lax type equations via the method of fractional powers:
\begin{equation}\label{eq:intro2}
\frac{dA(\partial)}{dt_{n,B}}
=
[B(\partial)^n_+,A(\partial)]
\,,
\end{equation}
where $B(\partial)$ is a $k$-th root of $A(\partial)$,
and $k,n$ are positive integers,
see \cite{DSKV15a} and Section \ref{sec:2.3} of the present paper.
Here, as usual, the subscript $+$ stands for the positive part
of a pseudodifferential operator.
Of course, in particular, for $A(\partial)=L_1(\partial)$ and $k=1$
(resp. $A(\partial)=L_{(M)}(\partial)$ and $k=M$),
equation \eqref{eq:intro2} is Sato's KP hierarchy \cite{Sat81}
(resp. Gelfand-Dickey's $M$-th KdV hierarchies \cite{GD76}).
See also \cite{Dic03}.

Furthermore, equations \eqref{eq:intro2}
are Hamiltonian with respect to the Poisson structure described above,
and they are bi-Hamiltonian if $A(\partial)$ is of bi-Adler type, see Section \ref{sec:6}.

The third basic property of Adler type pseudodifferential operators $A(\partial)$
is that it provides an infinite set of conserved densities
for the hierarchy \eqref{eq:intro2}:
\begin{equation}\label{eq:intro3}
h_{n,B}
=
\frac{-k}{n}\Res_\partial\tr B(\partial)^n
\,,\,\,n\in\mb Z_{>0}\,.
\end{equation}
Moreover, provided that $A(\partial)$ is bi-Adler,
these conserved densities satisfy the (generalized)
Lenard-Magri scheme, introduced in \cite{Mag78}.
See \cite{DSKV15a} and Sections \ref{sec:2.3} and \ref{sec:6}
of the present paper for details.

However, our approach allows one not only to recover the basic results
of the KP theory and $M$-th KdV theory,
including the matrix case,
as has been demonstrated in \cite{DSKV15a},
but to go far beyond that
in constructing new integrable (bi)Hamiltonian hierarchies of Lax type equations.
Of course, for this we need to construct new Adler type operators.

Our basic example is the family of $N\times N$-matrix differential operators
\begin{equation}\label{eq:intro4}
A_{S}(\partial)
=
\id_N\partial+\sum_{i,j=1}^Nq_{ji}E_{ij}+S^t
\,\in\Mat_{N\times N}\mc V[\partial]
\,,
\end{equation}
where $\mc V$ is the algebra of differential polynomials in the generators $q_{ij}$,
and $S^t$ is the transpose of a constant matrix $S\in\Mat_{N\times N}\mb F$.
It is not difficult to show,
see Example \ref{ex:affine-adler},
that the operator $A_S(\partial)$ is of Adler type
with respect to the following $\lambda$-bracket on $\mc V$,
making it the affine PVA over the Lie algebra $\mf{gl}_N$:
\begin{equation}\label{eq:intro5}
\{a_\lambda b\}_S
=
[a,b]+\tr(ab)\lambda+\tr(S^t[a,b])
\,\,,\,\,\,\,
a,b\in\mf{gl}_N
\,.
\end{equation}

Unfortunately, the operator $A_S(\partial)$ does not have ``good'' $K$-th roots,
see Appendix \ref{sec:app}.
The way out is provided by yet another remarkable property of a square matrix non-degenerate
Adler type operator $A(\partial)$:
its inverse is of Adler type as well, with respect to the negative of the $\lambda$-bracket
for $A(\partial)$ on $\mc V$.
It follows that the generalized quasideterminants of $A(\partial)$
are of Adler type with respect to the same $\lambda$-bracket as for $A(\partial)$,
see Section \ref{sec:3}.

Recall that a \emph{quasideterminant} of an invertible matrix $A$
over an associative (not necessarily commutative) unital ring
is the inverse (if it exists) of an entry of $A^{-1}$.
This notion was the key to the systematic development of linear algebra
over non-commutative associative rings
in a series of papers by Gelfand and Retakh, and their collaborators,
see \cite{GGRW05} for a review.
The basic idea of this work is that in linear algebra
one does not need determinants,
all one needs are inverse matrices!
(We find it surprising that,
nevertheless, in the above mentioned review the authors of a modern textbook in linear algebra
are criticized for the claim that ``determinants ... are of much less importance that they once were''.)

In the present paper we use the following generalization of quasideterminants:
if $I\in\Mat_{N\times M}\mb F$ and $J\in\Mat_{M\times N}\mb F$
are rectangular matrices of rank $M\leq N$,
the $(I,J)$ quasideterminant of an invertible matrix $A$ is $|A|_{IJ}=(JA^{-1}I)^{-1}$,
assuming that the $M\times M$ matrix $JA^{-1}I$ is invertible.

To demonstrate our new method of integrability,
consider, for example, the matrix $A_{\epsilon S}$ defined in \eqref{eq:intro4},
which is of Adler type with respect to the pencil of PVA $\lambda$-brackets
$\{\cdot\,_\lambda\,\cdot\}_\epsilon=\{\cdot\,_\lambda\,\cdot\}_0+\epsilon\{\cdot\,_\lambda\,\cdot\}_1$
on $\mc V$, defined by \eqref{eq:intro5} with $S$ replaced by $\epsilon S$.
Let $I\in\Mat_{N\times M}\mb F$ and $J\in\Mat_{M\times N}\mb F$
be rectangular matrices of rank $M\leq N$ such that $S=IJ$
($I$ and $J$ are uniquely defined up to a change of basis in $\mb F^M$).
Then, by the results of Section \ref{sec:3},
the generalized quasideterminant $|A|_{IJ}$ is 
a matrix pseudodifferential operator of bi-Adler type
with respect to the same pencil of PVA $\lambda$-brackets $\{\cdot\,_\lambda\,\cdot\}_\epsilon$.
As a consequence, 
we automatically get, by the above mentioned method of fractional powers (cf. \eqref{eq:intro3}),
a sequence of Hamiltonian functionals
$\tint h_n=-\frac1n\tint\Res_\partial\tr (|A(\partial)|_{IJ})^n$, $n\geq1$, $\tint h_0=0$,
which are in involution with respect to both PVA $\lambda$-brackets
$\{\cdot\,_\lambda\,\cdot\}_0$ and $\{\cdot\,_\lambda\,\cdot\}_1$,
they satisfy the Lenard-Magri recursion relation
in the subalgebra $\mc V_1\subset\mc V$ generated by the coefficients of $|A|_{IJ}$,
and therefore produce an integrable hierarchy of bi-Hamiltonian equations
$\frac{du}{dt_n}
=
\{{h_n}_\lambda u\}_0\big|_{\lambda=0}
=
\{{h_{n+1}}_\lambda u\}_1\big|_{\lambda=0}$, $n\geq0$, $u\in\mc V_1$.

Unfortunately it is a well established tradition
to call the Poisson structures corresponding to the $\lambda$-brackets
$\{\cdot\,_\lambda\,\cdot\}_0$ and $\{\cdot\,_\lambda\,\cdot\}_1$,
the second and the first Poisson structures respectively,
which disagrees with our notation.

In Section \ref{sec:7}
we describe this example in detail in the case of $\mf{gl}_2$,
deriving the first non-trivial equation of the corresponding hierarchy,
for all possible choices of the matrix $S$ in \eqref{eq:intro4}.

Next, we introduce in Section \ref{sec:8}
the notion of a dispersionless Adler type series, which allows us to construct
integrable Hamiltonian dispersionless hierarchy in a unified fashion.
For example, the symbols $L_M(z)$ and $L_{(M)}(z)$
of the ``generic'' pseudodifferential operators $L_M(\partial)$ and $L_{(M)}(\partial)$
given by \eqref{eq:intro1}
are dispersionless Adler type series.
The series $L_1(z)$ gives rise to the dispersionless KP hierarchy,
denoted by dKP, introduced in full generality in \cite{KG89},
of which the first non-trivial equation is Benney's equation for moments,
as described by Lebedev and Manin in \cite{LM79}.
As a result we endow the dKP hierarchy with an infinite series of bi-Poisson structures,
associated to $L_M(z)$ for every $M\geq1$.
The series $L_{(M)}(z)$ provides similar results for the dispersionless $M$-th KdV hierarchy.

Finally, in Section \ref{sec:9}
we develop further the notion of an Adler type pseudodifferential operator
with respect to a double PVA, introduced in \cite{DSKV15b}.

As an application of our method,
in our next paper \cite{DSKVfuture} we will 
construct an integrable hierarchy of bi-Hamiltonian equations
associated to any affine $\mc W$-algebra $\mc W(\mf{gl}_N,f)$,
where $f$ is any nilpotent element of the Lie algebra $\mf{gl}_N$.
%
This was previously known only for nilpotent elements $f$ 
of very specific type, see \cite{DSKV13} and references there.

Our base field $\mb F$ is a field of characteristic $0$.

\medskip

The first two authors would like to acknowledge
the hospitality of IHES, France,
where this work was completed in the summer of 2015.
The third author would like to acknowledge 
the help of the Italian government
in the evacuation from the earthquake in Nepal in April 2015,
as well as the friendship and hospitality of the Nepalese people,
and the hospitality of the University of Rome La Sapienza
during his subsequent visit in Rome.
The first author is supported by National FIRB grant RBFR12RA9W,
National PRIN grant 2012KNL88Y, and University grant C26A158K8A,
the second author is supported by an NSF grant,
and the third author is supported by an NSFC ``Research Fund for International
Young Scientists'' grant.

\section{Differential algebras, \texorpdfstring{$\lambda$}{lambda}-brackets and Poisson vertex algebras}\label{sec:2}

\subsection{Differential algebras and matrix pseudodifferential operators}\label{sec:2.-1}

By a differential algebra $\mc V$ we mean a unital commutative associative algebra
with a derivation $\partial$.
As usual, we denote by $\tint:\,\mc V\to\mc V/\partial\mc V$ the canonical quotient map.

Let $\mc V((\partial^{-1}))$ be the algebra of scalar pseudodifferential operators over $\mc V$.
Given $A(\partial)\in\mc V((\partial^{-1}))$,
we denote by $A(z)\in\mc V((z^{-1}))$ its symbol, 
obtained by replacing $\partial^n$ (on the right) by $z^n$.
The product $\circ$ on $\mc V((\partial^{-1}))$ can be defined in terms of symbols by
\begin{equation}\label{eq:prod-symbol}
(A\circ B)(z)=A(z+\partial)B(z)
\,.
\end{equation}
Here and further, we always expand an expression as $(z+\partial)^n$ in non-negative powers of $\partial$.
The formal adjoint of the pseudodifferential operator $A(\partial)=\sum_{n=-\infty}^Na_n\partial^n$ is
defined as 
$$
A^*(\partial)=\sum_{n=-\infty}^N(-\partial)^n\circ a_n
\,.
$$

Introduce an important notation which will be used throughout the paper.
Given a formal Laurent series $A(z)=\sum_{n=-\infty}^Na_nz^n\in\mc V((z^{-1}))$ 
and elements $b,c\in \mc V$, we let:
\begin{equation}\label{eq:notation}
A(z+x)\big(\big|_{x=\partial}b)c
=\sum_{n=-\infty}^Na_n\big((z+\partial)^nb\big)c
\,.
\end{equation}
For example, with this notation, the formal adjoint of a pseudodifferential
operator $A(\partial)\in\mc V((\partial^{-1}))$ has symbol
\begin{equation}\label{eq:adjoint}
A^*(z)=\big(\big|_{x=\partial}A(-z-x)\big)
\,.
\end{equation}

Let $\Mat_{M\times N}\mc V((\partial^{-1}))$ be the space of $M\times N$ matrix pseudodifferential 
operators over $\mc V$.
Let $A(\partial)=\big(A_{ij}(\partial)\big)_{\substack{{\scriptscriptstyle i=1,\dots,M}\\
\scriptscriptstyle j=1,\dots,N}}
\in\Mat_{M\times N}\mc V((\partial^{-1}))$
be a matrix pseudodifferential operator.
Its formal adjoint $A^*(\partial)\in\Mat_{N\times M}\mc V((\partial^{-1}))$ has entries
$(A^*)_{ij}(\partial)=(A_{ji})^*(\partial)$.
Let $I\subset\{1,\dots,M\}$ and $J\subset\{1,\dots,N\}$ be arbitrary subsets.
Given an $M\times N$ matrix $A=\big(A_{ij}\big)_{\substack{{\scriptscriptstyle i=1,\dots,M}\\ {\scriptscriptstyle j=1,\dots,N}}}$,
we denote by $A_{IJ}$ the corresponding $|I|\times|J|$ submatrix:
\begin{equation}\label{eq:submatrix}
A_{IJ}=\big(A_{ij}\big)_{i\in I,j\in J}
\,.
\end{equation}

Recall that the residue $\Res_zA(z)$ of a formal Laurent series $A(z)\in\mc V((z^{-1}))$
is defined as the coefficient of $z^{-1}$.
We introduce the following notation: 
\begin{equation}\label{eq:iota}
\iota_z(z-w)^{-1}=\sum_{n\in\mb Z_+}z^{-n-1}w^n
\,,
\end{equation}
i.e. $\iota_z$ denotes the geometric series expansion in the domain of large $z$
(and similarly, $\iota_w$ denotes the expansion for large $w$).
According to the previous convention (as in \eqref{eq:prod-symbol}), 
a Laurent series in $(z+\partial)$ is always expanded using $\iota_z$.
Similarly, if we have a Laurent series in $(z+\lambda)$ we omit the symbol $\iota_z$,
and $\iota_z$ or $\iota_w$ will only be used in Laurent series involving both variables $z$ and $w$.
Recall that, for a formal Laurent series $A(z)=\sum_{n=-\infty}^N a_nz^n\in\mc V((z^{-1}))$, we have
\begin{equation}\label{eq:positive}
\Res_z A(z)\iota_z(z-w)^{-1}
=
A(w)_+
:=
\sum_{n=0}^N a_nw^n
\,.
\end{equation}

We state here some simple facts about matrix pseudodifferential operators
which will be used later.
\begin{lemma}\phantomsection\label{lem:hn1}
\begin{enumerate}[(a)]
\item
Given two pseudodifferential operators $A(\partial),B(\partial)\in\mc V((\partial^{-1}))$, we have
\begin{equation}\label{eq:hn1a}
\Res_z A(z)B^*(\lambda-z)=\Res_zA(z+\lambda+\partial)B(z)
\,.
\end{equation}
\item
Given two matrix pseudodifferential operators 
$A(\partial),B(\partial)\in\Mat_{N\times N}\mc V((\partial^{-1}))$, we have
\begin{equation}\label{eq:hn1b}
\tint \Res_z \tr(A(z+\partial)B(z))=\tint \Res_z\tr(B(z+\partial)A(z))
\,.
\end{equation}
\end{enumerate}
\end{lemma}
\begin{proof}
First, for arbitrary $m,n\in\mb Z$, we have
$\Res_zz^m(z-\lambda)^n=\Res_z(z+\lambda)^mz^n\,\in\mb F[\lambda]((z^{-1}))$
(it reduces to the identity on binomial coefficients $\binom{k-n-1}{k}=(-1)^k\binom{n}{k}$ for $k\in\mb Z_+$).
Hence, by linearity, 
\begin{equation}\label{eq:hn1-pr2}
\Res_zA(z)B(z-\lambda)=\Res_zA(z+\lambda)B(z)
\,.
\end{equation}
for every $A(z),B(z)\in\mc V((z^{-1}))$.
On the other hand, in notation \eqref{eq:notation} equation \eqref{eq:hn1a} can be rewritten as
\begin{equation}\label{eq:hn1-pr1}
\Res_z A(z)\big(\big|_{x=\partial}B(z-\lambda-x)\big)=\Res_zA(z+\lambda+x)\big(\big|_{x=\partial}B(z)\big)
\,,
\end{equation}
which holds by \eqref{eq:hn1-pr2} (with $\lambda+x$ in place of $\lambda$).
Next, we prove part (b). We have
\begin{equation*}
\begin{split}
& \int\Res_z\tr(A(z+\partial)B(z))
=
\sum_{i,j=1}^N\int\Res_z A_{ij}(z+\partial)B_{ji}(z)
\\
& =
\sum_{i,j=1}^N\int\Res_z (B_{ji})^*(-z)A_{ij}(z)
=
\sum_{i,j=1}^N\int\Res_z \big(\big|_{x=\partial}B_{ji}(z-x)\big)A_{ij}(z)
\\
& =
\sum_{i,j=1}^N\int\Res_z B_{ji}(z+x)\big(\big|_{x=\partial}A_{ij}(z)\big)
=
\sum_{i,j=1}^N\int\Res_z B_{ji}(z+\partial)A_{ij}(z)
\\
& =
\int\Res_z \tr(B(z+\partial)A(z))
\,.
\end{split}
\end{equation*}
In the second equality we used equation \eqref{eq:hn1a} with $\lambda=0$, 
in the third equality we used \eqref{eq:adjoint},
and in the fourth equality we performed integration by parts.
\end{proof}

\subsection{\texorpdfstring{$\lambda$}{lambda}-brackets and Poisson vertex algebras}\label{sec:2.0}

Recall from \cite{BDSK09} that a $\lambda$-\emph{bracket} on the differential algebra $\mc V$ 
is a bilinear (over $\mb F$) map $\{\cdot\,_\lambda\,\cdot\}:\,\mc V\times\mc V\to\mc V[\lambda]$, 
satisfying the following
axioms ($a,b,c\in\mc V$):
\begin{enumerate}[(i)]
\item
sesquilinearity:
$\{\partial a_\lambda b\}=-\lambda\{a_\lambda b\}$,
$\{a_\lambda\partial b\}=(\lambda+\partial)\{a_\lambda b\}$;
\item
Leibniz rules (see notation \eqref{eq:notation}):
$$
\{a_\lambda bc\}=\{a_\lambda b\}c+\{a_\lambda c\}b
\,\,,\,\,\,\,
\{ab_\lambda c\}=\{a_{\lambda+x} c\} \big(\big|_{x=\partial}b\big)
+\{b_{\lambda+x} c\} \big(\big|_{x=\partial}a\big)
\,.
$$
\end{enumerate}
We say that $\mc V$ is a \emph{Poisson vertex algebra} (PVA) if the $\lambda$-bracket $\{\cdot\,_\lambda\,\cdot\}$
satisfies ($a,b,c\in\mc V$)
\begin{enumerate}[(i)]
\setcounter{enumi}{2}
\item
skewsymmetry:
$\{b_\lambda a\}=-\big(\big|_{x=\partial}\{a_{-\lambda-x} b\}\big)$;
\item
Jacobi identity:
$\{a_\lambda \{b_\mu c\}\}-\{b_\mu\{a_\lambda c\}\}
=\{\{a_\lambda b\}_{\lambda+\mu}c\}$.
\end{enumerate}
Recall \cite{BDSK09} that if $\{u_i\}_{i\in I}$ is a set of generators for the differential algebra $\mc V$,
then a $\lambda$-bracket on $\mc V$ is uniquely determined by $\{{u_i}_\lambda{u_j}\}$ for $i,j\in I$.
Moreover, if the skewsymmetry and Jacobi identity axioms hold on generators,
then $\mc V$ is a PVA.
If a differential algebra $\mc V$ is endowed with a $\lambda$-bracket
satisfying skewsymmetry and Jacobi identity,
we shall often say that $\mc V$ carries a PVA structure, or simply a Poisson structure.
\begin{example}\label{ex:affine-pva}
Let $\mc V$ be the algebra of differential polynomials over $\mf{gl}_N$.
To keep the notation separate,
we shall denote by $E_{ij}\in\Mat_{N\times N}\mb F$ the standard matrix with $1$ in position $(i,j)$ and $0$ everywhere else,
and by $q_{ij}\in\mf{gl}_N$ the same matrix, when viewed as an element of the differential algebra $\mc V$.
Then $\mc V=\mb F[q_{ij}^{(n)}\,|\,i,j=1,\dots,N,\,n\in\mb Z_+]$.
Let $S\in\Mat_{N\times N}\mb F$ be a fixed matrix
and let $S^t$ be its transpose.
The corresponding \emph{affine Poisson vertex algebra} $\mc V_S(\mf{gl}_N)$ is the differential algebra $\mc V$
endowed with the $\lambda$-bracket $\{\cdot\,_\lambda\,\cdot\}_S$ defined by
\begin{equation}\label{eq:affine-lb}
\{a_\lambda b\}_S
=
[a,b]+\tr(ab)\lambda+\tr(S^t[a,b])
\,,\qquad
a,b\in\mf{gl}_N\,,
\end{equation}
and (uniquely) extended to a PVA $\lambda$-bracket on $\mc V$
by the sesquilinearity axioms and the Leibniz rules.
\end{example}

\subsection{Properties of \texorpdfstring{$\lambda$}{lambda}-brackets}\label{sec:2.0b}

Later on we shall use the following results from \cite{DSK13},
which hold on an arbitrary differential algebra $\mc V$ 
with a $\lambda$-bracket $\{\cdot\,_\lambda\,\cdot\}$.
\begin{lemma}\label{lem:lambda-product}
For $A(\partial),B(\partial)\in\Mat_{N\times N}\mc V((\partial^{-1}))$,
$a\in\mc V$, and $i,j=1,\dots,N$, we have the following identities in $\mc V[\lambda]((z^{-1}))$
\begin{equation}\label{eq:lambda-product}
\begin{split}
& \{a_\lambda (A\circ B)_{ij}(z)\}
=
\sum_{k=1}^N \{a_\lambda A_{ik}(z+x)\}\big(\big|_{x=\partial}B_{kj}(z)\big)
\\
& \,\,\,\,\,\,\,\,\,\,\,\,\,\,\,\,\,\,\,\,\,\,\,\,\,\,\,\,\,\,\,\,\,\,\,\,\,\,\,\,\,\,\,\,\,
+\sum_{k=1}^N A_{ik}(z+\lambda+\partial)\{a_\lambda B_{kj}(z)\}
\,,\\
& \{(A\circ B)_{ij}(z)_\lambda a\}
=
\sum_{k=1}^N \{B_{kj}(z)_{\lambda+x}a\}\big(\big|_{x=\partial}(A^*)_{ki}(\lambda-z)\big)
\\
& \,\,\,\,\,\,\,\,\,\,\,\,\,\,\,\,\,\,\,\,\,\,\,\,\,\,\,\,\,\,\,\,\,\,\,\,\,\,\,\,\,\,\,\,\,
+\sum_{k=1}^N \{A_{ik}(z+x)_{\lambda+x}a\}\big(\big|_{x=\partial}B_{kj}(z)\big)
\,.
\end{split}
\end{equation}
\end{lemma}
\begin{proof}
It follows immediately from the Leibniz rules and the sesquilinearity axioms
(cf. the proof of \cite[Cor.3.11]{DSK13}).
\end{proof}
\begin{corollary}\label{cor:lambda-power}
For $A(\partial)\in\Mat_{N\times N}\mc V((\partial^{-1}))$,
$a\in\mc V$, $i,j=1,\dots,N$, and $n\in\mb Z_+$,
the following identities hold in $\mc V[\lambda]((z^{-1}))$
\begin{equation}\label{eq:lambda-power}
\begin{split}
& \{a_\lambda (A^n)_{ij}(z)\}
 =
\sum_{\ell=0}^{n-1}\sum_{h,k=1}^N
(A^{n-1-\ell})_{ih}(z+\lambda+\partial)
\{a_\lambda A_{hk}(z+x)\}
\big(\big|_{x=\partial}(A^\ell)_{kj}(z)\big)
\,,\\
& \{(A^n)_{ij}\!(z)_\lambda a\}
\\
&=
\sum_{\ell=0}^{n-1}\!\sum_{h,k=1}^N
\{A_{hk}(z+x)_{\lambda+x+y}a\}
\big(\big|_{x=\partial} (A^{n-1-\ell})_{kj}(z)\big)
\big(\big|_{y=\partial} (A^{*\ell})_{hi}(\lambda-z)\big)
\,.
\end{split}
\end{equation}
\end{corollary}
\begin{proof}
It follows immediately from Lemma \ref{lem:lambda-product}.
\end{proof}
\begin{lemma}\label{lem:lambda-inverse}
Let $A(\partial)\in\Mat_{N\times N}\mc V((\partial^{-1}))$
be an invertible matrix pseudodifferential operator,
and let $A^{-1}(\partial)\in\Mat_{N\times N}\mc V((\partial^{-1}))$
be its inverse.
Then, for $a\in\mc V$ and $i,j=1,\dots,N$, the following identities hold in $\mc V[\lambda]((z^{-1}))$
\begin{equation}\label{eq:lambda-inverse}
\begin{split}
& \big\{a_\lambda (A^{-1})_{ij}(z)\big\}
=
-\sum_{s,t=1}^N
(A^{-1})_{is}(z\!+\!\lambda\!+\!\partial)
\{a_\lambda A_{st}(z+x)\}
\Big(\Big|_{x=\partial} (A^{-1})_{tj}(z)\Big)
\,,\\
& \big\{(A^{-1})_{ij}(z)_\lambda a\big\}
\!=
\!-\!\sum_{s,t=1}^N\!
\{A_{st}(z\!+\!x)_{\lambda+x+y}a\}
\Big(\Big|_{x=\partial}\!\! (A^{-1})_{tj}(z)\Big)
\Big(\Big|_{y=\partial}\!\! (A^{*-1})_{si}(\lambda\!-\!z)\Big)
\,.
\end{split}
\end{equation}
(Note that $A^{*-1}=(A^{-1})^*$.)
\end{lemma}
\begin{proof}
This is the same as \cite[Lem.3.9]{DSK13} written in notation \eqref{eq:notation}.
\end{proof}
\begin{corollary}\label{cor:lambda-neg-power}
Let $A(\partial)\in\Mat_{N\times N}\mc V((\partial^{-1}))$
be an invertible matrix pseudodifferential operator.
Then, for $a\in\mc V$ and $i,j=1,\dots,N$, and $n\leq-1$,
the following identities hold in $\mc V[\lambda]((z^{-1}))$
\begin{equation}\label{eq:lambda-neg-power}
\begin{split}
& \{a_\lambda (A^n)_{ij}(z)\}
=
-\sum_{\ell=n}^{-1}\sum_{h,k=1}^N
(A^{n-1-\ell})_{ih}(z+\lambda+\partial)
\{a_\lambda A_{hk}(z+x)\}
\big(\big|_{x=\partial}(A^\ell)_{kj}(z)\big)
\,,\\
& \{(A^n)_{ij}\!(z)_\lambda a\}
\\
& =
-\sum_{\ell=n}^{-1}\sum_{h,k=1}^N
\{A_{hk}(z+x)_{\lambda+x+y}a\}
\big(\big|_{x=\partial} (A^{n-1-\ell})_{kj}(z)\big)
\big(\big|_{y=\partial} (A^{*\ell})_{hi}(\lambda-z)\big)
\,.
\end{split}
\end{equation}
\end{corollary}
\begin{proof}
It follows from Corollary \ref{cor:lambda-power} and Lemma \ref{lem:lambda-inverse}.
\end{proof}

\section{Operators of Adler type and their properties}\label{sec:2b}

\subsection{Adler type matrix pseudodifferential operators}\label{sec:2.1}

\begin{definition}[see {\cite{DSKV15a}}]\label{def:adler}
An $M\times N$ matrix pseudodifferential operator $A(\partial)$
over a differential algebra $\mc V$
is of \emph{Adler type} with respect to a $\lambda$-bracket $\{\cdot\,_\lambda\,\cdot\}$ on $\mc V$,
if, for every $(i,j),(h,k)\in\{1,\dots,M\}\times\{1,\dots,N\}$, we have
\begin{equation}\label{eq:adler}
\begin{split}
\{A_{ij}(z)_\lambda A_{hk}(w)\}
& = A_{hj}(w+\lambda+\partial)\iota_z(z\!-\!w\!-\!\lambda\!-\!\partial)^{-1}(A_{ik})^*(\lambda-z)
\\
& - A_{hj}(z)\iota_z(z\!-\!w\!-\!\lambda\!-\!\partial)^{-1}A_{ik}(w)
\,.
\end{split}
\end{equation}
(In \eqref{eq:adler} $(A_{ik})^*(\partial)$ denotes the formal adjoint of the scalar 
pseudodifferential operator $A_{ik}(\partial)$,
and $(A_{ik})^*(z)$ is its symbol.)
\end{definition}

Note that in \cite{DSKV15a} the Adler type operators are defined with the opposite sign 
in the RHS of \eqref{eq:adler}.
Obviously, this difference in sign does not affect any of the results in that paper,
but the present choice of sign seems to be more convenient.
\begin{remark}\label{rem:delta}
Note that the RHS of \eqref{eq:adler} is in fact regular at $z=w+\lambda+\partial$,
namely we can replace the expansion $\iota_z$ by the expansion $\iota_w$.
To see this, just observe that
$$
\iota_z(z-w-\lambda-\partial)^{-1}-\iota_w(z-w-\lambda-\partial)^{-1}=\delta(z-w-\lambda-\partial)
\,,
$$
the formal $\delta$-function, and use the following properties of the formal delta function:
$$
\begin{array}{l}
\displaystyle{
\vphantom{\Big(}
A(z)\delta(z-w-\lambda-\partial)
=
A(w+\lambda+\partial)\delta(z-w-\lambda-\partial)
\,,} \\
\displaystyle{
\vphantom{\Big(}
\delta(z-w-\lambda-\partial)A^*(\lambda-z)
=
\delta(z-w-\lambda-\partial)A(w)
\,,}
\end{array}
$$
for any scalar pseudodifferential operator $A(\partial)$.
\end{remark}
\begin{example}\label{ex:constant}
Any $M\times N$ matrix with entries in $\mb F$
is a pseudodifferential operator of Adler type
for any $\lambda$-bracket on $\mc V$.
\end{example}
\begin{example}\label{ex:affine-adler}
Consider the affine PVA $\mc V_S(\mf{gl}_N)$ from Example \ref{ex:affine-pva}.
Consider the following $N\times N$ matrix differential operator over $\mc V_S(\mf{gl}_N)$:
\begin{equation}\label{eq:mcA}
A_S(\partial)=\id_N\partial+\sum_{i,j=1}^N q_{ji}E_{ij}+S^t
=
\left(\begin{array}{llll}
\partial+q_{11}&q_{21}&\dots&q_{N1} \\
q_{12}&\partial+q_{22}&\dots&q_{N2} \\
\vdots&\vdots&\ddots&\vdots \\
q_{1N}&q_{2N}&\dots&\partial+q_{NN}
\end{array}\right)
+S^t
\,.
\end{equation}
We claim that the matrix differential operator $\mc A_S(\partial)$ is of Adler type
with respect to the affine $\lambda$-bracket \eqref{eq:affine-lb}.
Indeed, the RHS of \eqref{eq:adler} is, for the matrix \eqref{eq:mcA},
\begin{equation*}
\begin{split}
& A_{hj}(w+\lambda+\partial)\iota_z(z\!-\!w\!-\!\lambda\!-\!\partial)^{-1}(A_{ik})^*(\lambda-z)
- A_{hj}(z)\iota_z(z\!-\!w\!-\!\lambda\!-\!\partial)^{-1}A_{ik}(w)
\\
& =
(\delta_{hj}(w+\lambda+\partial)+q_{jh}+s_{jh})\iota_z(z-w-\lambda-\partial)^{-1}(\delta_{ik}(z-\lambda)+q_{ki}+s_{ki})
\\
& 
\,\,\,\,\,\,\,
- (\delta_{hj}z+q_{jh}+s_{jh})\iota_z(z-w-\lambda-\partial)^{-1}(\delta_{ik}w+q_{ki}+s_{ki})
\\
& =
\delta_{ik}(q_{jh}+s_{jh})-\delta_{hj}(q_{ki}+s_{ki})+\lambda\delta_{hj}\delta_{ik}
\\
& =
\{{q_{ji}}_\lambda{q_{kh}}\}_S
=\{A_{ij}(z)_\lambda A_{hk}(w)\}_S
\,.
\end{split}
\end{equation*}
\end{example}
\begin{example}\label{ex:affine-generic}
Let $M$ be a positive integer
and let $\mc V$ be the algebra of differential polynomials
in the infinitely many variables $u_{j;ab}$, 
where $1\leq a,b\leq N$ and $-\infty<j< M$
(resp. in the finitely many variables
$u_{j;ab}$, where $1\leq a,b\leq N$ and $0\leq j< M$).
Consider the following ``generic'' $N\times N$ matrix pseudodifferential 
(resp. differential) operators:
\begin{equation}\label{eq:gener1}
L_{MN}(\partial)
=
\sum_{j=-\infty}^M U_j\partial^j
\qquad\Big(\text{resp. }\,\,
L_{(MN)}(\partial)
=
\sum_{j=0}^M U_j\partial^j
\Big)
\,\,\text{ with }\,\,
U_M=\id_N
\,,
\end{equation}
where $U_j=\big(u_{j;ab}\big)_{a,b=1}^N\in\Mat_{N\times N}\mc V$ for every $j<M$.
We claim that, for every $S\in\Mat_{N\times N}\mb F$,
the one parameter family of operators $A(\partial)=L_{MN}(\partial)+\epsilon S$ 
(resp. $A(\partial)=L_{(MN)}+\epsilon S$), $\epsilon\in\mb F$, is of Adler type
with respect to a pencil of $\lambda$-brackets
$\{\cdot\,_\lambda\,\cdot\}_\epsilon=\{\cdot\,_\lambda\,\cdot\}_0+\epsilon\{\cdot\,_\lambda\,\cdot\}_1$
on $\mc V$.
Indeed, in this case the LHS of \eqref{eq:adler}
becomes a Laurent series (resp. polynomial) in $z$ and $w$ of highest degree $N-1$
in $z$ and in $w$,
whose coefficients are the unknown $\lambda$-brackets
$\{{u_{i;ab}}_\lambda{u_{j;cd}}\}_\epsilon$, which are subject only to the
conditions coming from the skewsymmetry axiom of Section \ref{sec:2.0},
and to the assumption of being linear in $\epsilon$.
On the other hand, the RHS of \eqref{eq:adler}
has obviously degree at most $N-1$ in $z$,
and therefore it has also degree at most $N-1$ in $w$ due to 
Remark \ref{rem:delta},
it is a quadratic expression in $\epsilon$,
and the coefficient of $\epsilon^2$ is immediately seen to be zero.
In conclusion, the pencil of $\lambda$-brackets 
$\{\cdot\,_\lambda\,\cdot\}_\epsilon$ is defined on the generators of $\mc V$
by \eqref{eq:adler}, and it extends to the algebra $\mc V$
by the sesquilinearity axioms and the Leibniz rules, \cite{BDSK09}.
We point out that, as we shall show in \cite{DSKVfuture},
$L_{(MN)}(\partial)$ is of Adler type with respect to the PVA structure
of the $\mc W$-algebra associated to the Lie algebra $\mf{gl}_{NM}$
and rectangular nilpotent $f\in\mf{gl}_{NM}$ with $N$ Jordan blocks of size $M$.
\end{example}

The relation between Adler type operators and Poisson vertex algebras
is given by the following result, 
which is an immediate consequence of \cite[Lem.4.1]{DSKV15a} (see also \cite[Rem.2.6]{DSKV15a}):
\begin{theorem}\label{thm:main-adler}
Let $A(\partial)\in\Mat_{M\times N}\mc V((\partial^{-1}))$ be an $M\times N$-matrix pseudodifferential operator
of Adler type with respect to a $\lambda$-bracket $\{\cdot\,_\lambda\,\cdot\}$ on $\mc V$,
and assume that the coefficients of the entries of the matrix $A(\partial)$
generate $\mc V$ as a differential algebra.
Then $\mc V$ is a Poisson vertex algebra (with $\lambda$-bracket $\{\cdot\,_\lambda\,\cdot\}$).
\end{theorem}

\subsection{Properties of Adler type pseudodifferential operators}\label{sec:2.4}

\begin{proposition}\label{prop:properties-adler}
Let $\mc V$ be a differential algebra with a $\lambda$-bracket $\{\cdot\,_\lambda\,\cdot\}$.
Let $A(\partial)$ be an $M\times N$ matrix pseudodifferential operator of Adler type
with respect to the $\lambda$-bracket of $\mc V$.
Then the following properties hold:
\begin{enumerate}[(a)]
\item 
For every subsets $I\subset\{1,\dots,M\}$, $J\subset\{1,\dots,N\}$,
the corresponding $|I|\times|J|$ submatrix $A_{IJ}$ is of Adler type 
(with respect to the same $\lambda$-bracket $\{\cdot\,_\lambda\,\cdot\}$).
In particular, every entry $A_{ij}(\partial)$ of the matrix $A(\partial)$
is a scalar pseudodifferential operator of Adler type.
\item
The formal adjoint $N\times M$ matrix pseudodifferential operator $A^*(\partial)$
is of Adler type 
(with respect to the same $\lambda$-bracket $\{\cdot\,_\lambda\,\cdot\}$).
\item
Assuming that $M=N$ and that the square matrix pseudodifferential operator $A(\partial)$ is invertible
in the algebra $\Mat_{N\times N}\mc V((\partial^{-1}))$,
the inverse matrix $A(\partial)^{-1}$ is of Adler type
with respect to the opposite $\lambda$-bracket \, $-\{\cdot\,_\lambda\,\cdot\}$.
\item
If $A(\partial)\in\Mat_{M\times N}\mc V((\partial^{-1}))$ 
and $B(\partial)\in\Mat_{N\times P}\mc V((\partial^{-1}))$
are both matrices of Adler type with respect to the $\lambda$-bracket $\{\cdot\,_\lambda\,\cdot\}$,
and if $\{A_{ij}(z)_\lambda B_{hk}(w)\}=0$ for all $i,j,h,k$,
then the product $A(\partial)\circ B(\partial)\in\Mat_{M\times P}\mc V((\partial^{-1}))$
is a matrix of Adler type.
In particular, if $S\in\Mat_{P\times M}\mb F$ and $T\in\Mat_{N\times Q}\mb F$
are constant matrices,
then the matrix $SA(\partial)T\in\Mat_{P\times Q}\mc V((\partial^{-1}))$
is of Adler type.
\end{enumerate}
\end{proposition}
\begin{proof}
Part (a) is an obvious consequence of the definition \eqref{eq:adler} of Adler type matrix
pseudodifferential operators.
Next, let us prove part (b). 
By the sesquilinearity assumptions, we have, using notation \eqref{eq:notation}
\begin{equation}\label{eq:20150615-1}
\{(A^*)_{ij}(z)_\lambda(A^*)_{hk}(w)\}
=
\Big(\Big|_{x=\partial}\{A_{ji}(\lambda-z)_\lambda A_{kh}(-w-\lambda-x)\}\Big)
\,.
\end{equation}
We can now use condition \eqref{eq:adler} on the matrix $A(\partial)$
to rewrite the RHS of \eqref{eq:20150615-1} as
\begin{equation*}
\begin{split}
& 
\Big(\Big|_{x=\partial}
- A_{ki}(-w-x+\partial)\iota_z(z-w-\lambda-x+\partial)^{-1}(A_{jh})^*(z)
\\
& + A_{ki}(\lambda-z)\iota_z(z-w-\lambda-x+\partial)^{-1}A_{jh}(-w-\lambda-x)
\Big)
\\
& =
- (A^*)_{hj}(z) \iota_z(z-w-\lambda-\partial)^{-1} (A^*)_{ik}(w)
\\
& + (A^*)_{hj}(w+\lambda+\partial) \iota_z(z-w-\lambda-\partial)^{-1} ((A^*)_{ik})^*(\lambda-z)
\,,
\end{split}
\end{equation*}
proving that the matrix pseudodifferential operator $A^*(\partial)$ is of Adler type.

Next, we prove part (c).
By Lemma \ref{lem:lambda-inverse}, we have
\begin{equation}\label{eq:20150615-2}
\begin{split}
& \{(A^{-1})_{ij}(z)_\lambda (A^{-1})_{hk}(w)\}
\\
& =
\sum_{p,q,s,t=1}^N
(A^{-1})_{hs}(w+\lambda+\partial)
\{A_{pq}(z+x_1)_{\lambda+x_1+x_3}A_{st}(w+x_2)\}
\\
& \,\,\,\,\,\,\times
\Big(\Big|_{x_1=\partial} (A^{-1})_{qj}(z)\Big)
\Big(\Big|_{x_2=\partial} (A^{-1})_{tk}(w)\Big)
\Big(\Big|_{x_3=\partial} (A^{*-1})_{pi}(\lambda-z)\Big)
\,.
\end{split}
\end{equation}
We then apply condition \eqref{eq:adler} on the matrix $A(\partial)$
to rewrite the RHS of \eqref{eq:20150615-2} as
\begin{equation}\label{eq:20150615-3}
 \begin{split}
& \sum_{p,q,s,t=1}^N
(A^{-1})_{hs}(w+\lambda+\partial)
\\
& \times
\Big(
A_{sq}(w\!+\!\lambda\!+\!x_1\!+\!x_2\!+\!x_3\!+\!\partial)
\iota_z(z\!-\!w\!-\!\lambda\!-\!x_2\!-\!x_3\!-\!\partial)^{-1}
(A_{pt})^*(\lambda\!+\!x_3\!-\!z)
\\
&\,\,\,\,\,\,\,\,\,\,\,\,\,\,\,\,\,\,\,\,\,\,\,\,\,\,\,\,\,\,\,\,\,\,\,\,\,\,\,\,\,\,\,\,\,
- A_{sq}(z\!+\!x_1)
\iota_z(z\!-\!w\!-\!\lambda\!-\!x_2\!-\!x_3\!-\!\partial)^{-1}
A_{pt}(w+x_2)
\Big)
\\
& \times
\Big(\Big|_{x_1=\partial} (A^{-1})_{qj}(z)\Big)
\Big(\Big|_{x_2=\partial} (A^{-1})_{tk}(w)\Big)
\Big(\Big|_{x_3=\partial} (A^{*-1})_{pi}(\lambda-z)\Big)
\,.
\end{split}
\end{equation}
By definition of inverse matrix we have
\begin{equation*}
\begin{split}
& \sum_{s=1}^N
(A^{-1})_{hs}(w+\lambda+\partial)
\circ A_{sq}(w+\lambda+\partial)
=\delta_{qh}
\,,\\
& \sum_{p=1}^N
(A_{pt})^*(\lambda\!+\!\partial\!-\!z)
(A^{*-1})_{pi}(\lambda-z)
=\delta_{ti}
\,,\\
& \sum_{q=1}^N
A_{sq}(z\!+\!\partial)
(A^{-1})_{qj}(z)
=\delta_{sj}
\,,\,\,
\sum_{t=1}^N
A_{pt}(w+\partial)(A^{-1})_{tk}(w)
=\delta_{pk}
\,.
\end{split}
\end{equation*}
Hence, the RHS of \eqref{eq:20150615-3} can be rewritten as
\begin{equation*}
\begin{split}
& (A^{-1})_{hj}(z)
\iota_z(z-w-\lambda-\partial)^{-1}
(A^{-1})_{ik}(w)
\\
& -
(A^{-1})_{hj}(w+\lambda+\partial)
\iota_z(z-w-\lambda-\partial)^{-1}
(A^{*-1})_{ki}(\lambda-z)\,,
\end{split}
\end{equation*}
proving that $A^{-1}(\partial)$ is of Adler type with respect to the opposite $\lambda$-bracket
\,$-\{\cdot\,_\lambda\,\cdot\}$.

Finally, the proof of (d) is analogous to the proof of \cite[Prop.2.14]{DSKV15a}.
We reproduce here the argument for completeness.
By Lemma \ref{lem:lambda-product} and the assumption that the $\lambda$-bracket
between entries of $A$ and $B$ is zero, we have
\begin{equation}\label{eq:proofd-1}
\begin{split}
& \{(A\circ B)_{ij}(z)_\lambda(A\circ B)_{hk}(w)\}
\\
& =
\sum_{s,t=1}^N
\{A_{is}(z+x)_{\lambda+x} A_{ht}(w+y)\}
\big(\big|_{x=\partial}B_{sj}(z)\big)
\big(\big|_{y=\partial}B_{tk}(w)\big)
\\
& +
\sum_{s,t=1}^N
A_{ht}(w+\lambda+\partial)
\{B_{sj}(z)_{\lambda+x} B_{tk}(w)\}
\big(\big|_{x=\partial}(A_{is})^*(\lambda-z)\big)
\,.
\end{split}
\end{equation}
By the Adler condition \eqref{eq:adler},
the first term in the RHS of \eqref{eq:proofd-1} is equal to
\begin{equation}\label{eq:proofd-2}
\begin{split}
& \sum_{s,t=1}^N
A_{hs}(w+\lambda+\partial)B_{sj}(z)\iota_z(z-w-\lambda-\partial)^{-1}(A_{it})^*(\lambda-z)B_{tk}(w)
\\
& -
(A\circ B)_{hj}(z)\iota_z(z-w-\lambda-\partial)^{-1}(A\circ B)_{ik}(w)
\,,
\end{split}
\end{equation}
while the second term in the RHS of \eqref{eq:proofd-1} is
\begin{equation}\label{eq:proofd-3}
\begin{split}
& (A\circ B)_{hj}(w+\lambda+\partial)\iota_z(z-w-\lambda-\partial)^{-1}((A\circ B)_{ik})^*(\lambda-z)
\\
& - \sum_{s,t=1}^N
A_{ht}(w+\lambda+\partial)B_{tj}(z)\iota_z(z-w-\lambda-\partial)^{-1}(A_{is})^*(\lambda-z)B_{sk}(w)
\,.
\end{split}
\end{equation}
Combining \eqref{eq:proofd-2} and \eqref{eq:proofd-3},
we conclude that $A\circ B$ satisfies the Adler condition \eqref{eq:adler},
as claimed.
\end{proof}

\section{Quasideterminants of operators of Adler type}\label{sec:3}

\subsection{A brief theory of quasideterminants}\label{sec:3.1}

In this section we review the notion of quasideterminants of a matrix over an arbitrary 
unital associative ring, and we prove some of their properties.
For an extended review on the subject, see \cite{GGRW05}.

Let $R$ be a unital associative ring.
Let $A\in\Mat_{N\times N}R$ be an $N\times N$ matrix with entries in $R$.
Let $I,J\subset\{1,\dots,N\}$ be subsets of the same cardinality $|I|=|J|=M$.
\begin{definition}\label{def:quasidet}
The $(I,J)$-\emph{quasideterminant} of $A$ is
defined as (recall the notation \eqref{eq:submatrix})
\begin{equation}\label{eq:def-quasidet}
|A|_{IJ}:=((A^{-1})_{JI})^{-1}
\,\,\in\Mat_{M\times M}R
\,,
\end{equation}
assuming that the RHS makes sense, i.e. that the $N\times N$ matrix $A$ is invertible,
and that the $M\times M$ matrix $(A^{-1})_{JI}$ is also invertible.
(The Gelfand-Retakh quasideterminant was originally defined in the special case $M=1$.)

\end{definition}
When the quasideterminant is defined, it is possible to write a simple formula for it.
This is given by the following
\begin{proposition}\label{prop:quasidet}
Assume that $A\in\Mat_{N\times N}R$ is invertible and that 
the submatrix $A_{I^cJ^c}\in\Mat_{(N-M)\times(N-M)}$ is invertible,
where $I^c$ and $J^c$ denote the complement of the sets $I$ and $J$ in $\{1,\dots,N\}$.
Then the $(I,J)$-quasideterminant of $A$ exists, and it is given by the following formula:
\begin{equation}\label{eq:form-quasidet}
|A|_{IJ}=A_{IJ}-A_{IJ^c}(A_{I^cJ^c})^{-1}A_{I^cJ}
\,.
\end{equation}
\end{proposition}
\begin{proof}
Without loss of generality, we can assume that $I=J=\{1,\dots,M\}\subset\{1,\dots,N\}$.
In this case, the matrix $A$ can be written in blocks as follows:
$$
A=\left(\begin{array}{ll}
A_{IJ} & A_{IJ^c} \\
A_{I^cJ} & A_{I^cJ^c}
\end{array}\right)
\,.
$$
Accordingly, the inverse matrix has the block form
$$
A^{-1}=\left(\begin{array}{ll}
(A^{-1})_{JI} & (A^{-1})_{JI^c} \\
(A^{-1})_{J^cI} & (A^{-1})_{J^cI^c}
\end{array}\right)
\,.
$$
Under the assumptions that $A$ and $A_{I^cJ^c}$ are invertible,
it is easy to find a formula for the four blocks of $A^{-1}$:
\begin{equation*}
\begin{split}
& (A^{-1})_{JI}=
(A_{IJ}-A_{IJ^c}(A_{I^cJ^c})^{-1}A_{I^cJ})^{-1}
\,,\\
& (A^{-1})_{JI^c}=
-(A^{-1})_{JI}A_{IJ^c}(A_{I^cJ^c})^{-1}
\,,\\
& (A^{-1})_{J^cI}=
-(A_{I^cJ^c})^{-1}A_{I^cJ}(A^{-1})_{JI}
\,,\\
& (A^{-1})_{J^cI^c}=
(A_{I^cJ^c})^{-1}
+(A_{I^cJ^c})^{-1}A_{I^cJ}(A^{-1})_{JI}A_{IJ^c}(A_{I^cJ^c})^{-1}
\,.
\end{split}
\end{equation*}
The claim follows.
\end{proof}
Quasideterminant satisfy the so-called \emph{hereditary property},
which can be stated as follows:
\begin{proposition}\label{prop:quasidet_hereditary}
Let $I_1\subset I\subset\{1,\dots,N\}$ and $J_1\subset J\subset\{1,\dots,N\}$
be subsets such that $|I_1|=|J_1|$ and $|I|=|J|$.
Assuming that the $(I,J)$-quasideterminant and the $(I_1,J_1)$-quasideterminant of $A$ exist, 
we have
\begin{equation}\label{eq:form-quasidet_hereditary}
|A|_{I_1J_1}=||A|_{IJ}|_{I_1J_1}
\,.
\end{equation}
\end{proposition}
\begin{proof}
By the definition \eqref{eq:def-quasidet} of quasideterminant we have
$$
||A|_{IJ}|_{I_1J_1}=
(((((A^{-1})_{JI})^{-1})^{-1})_{J_1I_1})^{-1}=
((A^{-1})_{J_1I_1})^{-1}=|A|_{I_1J_1}\,.
$$
\end{proof}

\subsection{Generalized quasideterminants}\label{sec:3.1b}

We can generalize the notion of quasideterminant as follows.
Let, as before, $A\in\Mat_{N\times N}R$,
and let $I\in\Mat_{N\times M}R$, $J\in\Mat_{M\times N}R$,
for some $M\leq N$.
\begin{definition}\label{def:gen-quasidet}
The $(I,J)$-\emph{quasideterminant} of $A$ is
\begin{equation}\label{eq:gen-quasidet}
|A|_{IJ}
=
(JA^{-1}I)^{-1}\,\in\Mat_{M\times M}R
\,,
\end{equation}
assuming that the RHS makes sense, i.e. that $A$ is invertible in $\Mat_{N\times N}R$
and that $JA^{-1}I$ is invertible in $\Mat_{M\times M}R$.
\end{definition}
Note that Definition \ref{def:gen-quasidet} is a generalization
of Definition \ref{def:quasidet}.
Indeed, given the subsets $I=\{i_1,\dots,i_M\}$ and 
$J=\{j_1,\dots,j_M\}$ of $\{1,\dots,N\}$,
then the quasideterminant $|A|_{IJ}$ defined by \eqref{eq:def-quasidet}
coincides with the quasideterminant $|A|_{C_I,R_J}$ defined by \eqref{eq:gen-quasidet},
where $C_I$ is the $N\times M$ matrix with $1$ in position $(i_\ell,\ell)$, for $\ell=1,\dots,M$,
and zero otherwise,
and $R_J$ is the $M\times N$ matrix with $1$ in position $(\ell,j_\ell)$, for $\ell=1,\dots,M$,
and zero otherwise.
The following result, though very simple,
will be very useful in our theory.
\begin{theorem}\label{thm:main-quasidet}
Let $A\in\Mat_{N\times N}R$, $I\in\Mat_{N\times M}R$, $J\in\Mat_{M\times N}R$, for some $M\leq N$.
Assume that the $(I,J)$-quasideterminant $|A|_{IJ}$ exists
and that the matrix $A+IJ\in\Mat_{N\times N}R$ is invertible.
Then, the $(I,J)$-quasideterminant of $A+IJ$ exists, and it is given by
\begin{equation}\label{eq:main-quasidet}
|A+IJ|_{IJ}=|A|_{IJ}+\id_{M}
\,.
\end{equation}
\end{theorem}
\begin{proof}
If we multiply $A+IJ$ 
by $J(A+IJ)^{-1}$ on the left and by $A^{-1}I|A|_{IJ}$ on the right, we get
\begin{equation}\label{explan1}
J(A+IJ)^{-1}(A+IJ)A^{-1}I|A|_{IJ}
=
(JA^{-1}I)|A|_{IJ}
=\id_M
\,,
\end{equation}
by definition of quasideterminant.
On the other hand, we can compute the same expression differently,
as follows
\begin{equation}\label{explan2}
\begin{array}{l}
\displaystyle{
\vphantom{\Big(}
J(A+IJ)^{-1}(A+IJ)A^{-1}I\,|A|_{IJ}
=
J(A+IJ)^{-1}I(\id_M+JA^{-1}I)|A|_{IJ}
} \\
\displaystyle{
\vphantom{\Big(}
=
J(A+IJ)^{-1}I(|A|_{IJ}+\id_M)
\,,}
\end{array}
\end{equation}
again by definition of quasideterminant.
Comparing \eqref{explan1} and \eqref{explan2}, we get
$$
J(A+IJ)^{-1}I(\id_M+|A|_{IJ})
=
\id_M
\,,
$$
i.e. $\id_M+|A|_{IJ}$ is a right inverse of $J(A+IJ)^{-1}I$.
Similarly, 
if we multiply $A+IJ$ 
by $|A|_{IJ}JA^{-1}$ on the left and by $(A+IJ)^{-1}I$ on the right, we get
$$
(\id_M+|A|_{IJ})J(A+IJ)^{-1}I
=
\id_M
\,,
$$
i.e. $\id_M+|A|_{IJ}$ is also a left inverse of $J(A+IJ)^{-1}I$.
In conclusion, $\id_M+|A|_{IJ}=(J(A+IJ)^{-1}I)^{-1}=|A+IJ|_{IJ}$,
proving the claim.
\end{proof}

\subsection{Adler type operators and quasideterminants}\label{sec:3.2}

\begin{proposition}\label{thm:quasidet-adler}
Let $\mc V$ be a differential algebra with a $\lambda$-bracket $\{\cdot\,_\lambda\,\cdot\}$.
Let $A(\partial)\in\Mat_{N\times N}\mc V((\partial^{-1}))$ be a matrix 
pseudodifferential operator of Adler type with respect to the $\lambda$-bracket of $\mc V$.
Then, for every $I\in\Mat_{N\times M}\mb F$
and $J\in\Mat_{M\times N}\mb F$ with $M\leq N$,
the quasideterminant $|A(\partial)|_{IJ}$
is a matrix pseudodifferential operator of Adler type,
provided that it exists.
\end{proposition}
\begin{proof}
It is an obvious consequence of Proposition \ref{prop:properties-adler}(c) and (d)
and of the definition \eqref{eq:gen-quasidet} of quasideterminant.
\end{proof}
\begin{remark}\label{rem:daniele}
For $I,J\subset\{1,\dots,N\}$ of the same cardinality,
one can show, by direct computations similar to the ones used 
to prove Proposition \ref{prop:properties-adler}(c),
that $A_{IJ}(\partial)-A_{IJ^c}(\partial)\circ(A_{I^cJ^c}(\partial))^{-1}\circ A_{I^cJ}(\partial)$
is a matrix of Alder type,
even if the quasideterminant $|A(\partial)|_{IJ}$ does not exist,
assuming of course that $A_{I^cJ^c}(\partial)$ is invertible.
We omit the proof of this fact.
\end{remark}

\section{Integrable hierarchy associated to a matrix pseudodifferential operator of Adler type}\label{sec:2.3}

Let $\mc V$ be a Poisson vertex algebra with $\lambda$-bracket $\{\cdot\,_\lambda\,\cdot\}$.
We have the corresponding Lie algebra structure on $\mc V/\partial\mc V$
with Lie bracket 
\begin{equation}\label{eq:lie-br}
\{\tint f,\tint g\}=\tint\{f_\lambda g\}\big|_{\lambda=0}
\,,
\end{equation}
and a representation of the Lie algebra $\mc V/\partial\mc V$ on $\mc V$
given by the action 
\begin{equation}\label{eq:lie-act}
\{\tint f,g\}=\{f_\lambda g\}\big|_{\lambda=0}
\,.
\end{equation}
Recall that the basic problem in the theory of integrability
is to construct an infinite linearly independent sequence of elements 
$\tint h_n\in\mc V/\partial\mc V,\,n\in\mb Z_+$,
called Hamiltonian functionals, which are in involution, i.e. such that
$$
\{\tint h_m,\tint h_n\}=0
\,\text{ for all } m,n\in\mb Z_+
\,.
$$
In this case we obtain a hierarchy of compatible Hamiltonian equations
$$
\frac{du}{dt_n}=\{\tint h_n,u\}
\,,\,\, u\in\mc V
\,.
$$
In the present section we show that,
given a matrix pseudodifferential operator of Adler type,
it is possible to construct a sequence of Hamiltonian functionals in involution
(the question of their linearly independence is treated separately).
This is stated in the following:
\begin{theorem}\label{thm:hn}
Let $\mc V$ be a differential algebra with a $\lambda$-bracket $\{\cdot\,_\lambda\,\cdot\}$.
Let $A(\partial)\in\Mat_{N\times N}\mc V((\partial^{-1}))$
be a matrix pseudodifferential operator of Adler type 
with respect to the $\lambda$-bracket $\{\cdot\,_\lambda\,\cdot\}$,
and assume that $A(\partial)$ is invertible in $\Mat_{N\times N}\mc V((\partial^{-1}))$.
For $B(\partial)\in\Mat_{N\times N}\mc V((\partial^{-1}))$
a $K$-th root of $A$ (i.e. $A(\partial)=B(\partial)^K$ for $K\in\mb Z\backslash\{0\}$)
define the elements $h_{n,B}\in\mc V$, $n\in\mb Z$, by
\begin{equation}\label{eq:hn}
h_{n,B}=
\frac{-K}{|n|}
\Res_z\tr(B^n(z))
\text{ for } n\neq0
\,,\,\,
h_0=0\,.
\end{equation}
Then: 
\begin{enumerate}[(a)]
\item
All the elements $\tint h_{n,B}$ are Hamiltonian functionals in involution:
\begin{equation}\label{eq:invol}
\{\tint h_{m,B},\tint h_{n,C}\}=0
\,\text{ for all } m,n\in\mb Z,\,
B,C \text{ roots of } A
\,.
\end{equation}
\item
The corresponding compatible hierarchy of Hamiltonian equations satisfies
\begin{equation}\label{eq:hierarchy}
\frac{dA(z)}{dt_{n,B}}
=
\{\tint h_{n,B},A(z)\}
=
[(B^n)_+,A](z)
\,,\,\,n\in\mb Z,\,
B \text{ root of } A
\end{equation}
(in the RHS we are taking the symbol of the commutator of matrix pseudodifferential operators),
and the Hamiltonian functionals $\tint h_{n,C}$, $n\in\mb Z_+$, $C$ root of $A$,
are integrals of motion of all these equations.
\end{enumerate}
\end{theorem}
\begin{remark}
Since, by assumption, $A(\partial)$ is invertible in $\Mat_{N\times N}\mc V((\partial^{-1}))$,
then, for $K\geq1$, a $K$-th root $B(\partial)$ of $A(\partial)$
is invertible as well, its inverse being $B^{-1}=A^{-1}\circ B^{K-1}$.
Hence, \eqref{eq:hn} makes sense for all $n\in\mb Z$.
On the other hand, the sequence \eqref{eq:hn} for $B$
is related to the same sequence for $B^{-1}$ by the following simple relation:
$h_{-n,B}=h_{n,B^{-1}}$ for all $n\in\mb Z_+$.
\end{remark}
\begin{remark}\label{rem:negative-n}
One can state the same Theorem \ref{thm:hn} without the assumption that $A$ 
(and therefore $B$) is invertible,
at the price of assuming that $K\geq1$,
and of restricting the sequence $h_n$ in \eqref{eq:hn} to $n\in\mb Z_+$.
Moreover, since the proof of \eqref{eq:invol} and \eqref{eq:hierarchy} 
is based on Lemma \ref{lem:hn2},
one needs to restrict equation \eqref{eq:invol} to $m\geq K$ and $n\geq L$,
where $B^K=C^L=A$,
and equation \eqref{eq:hierarchy} to $n\geq K$.
\end{remark}
\begin{remark}\label{rem:V1}
Let $\mc V_1\subset\mc V$ be the differential subalgebra generated by the coefficients
of $A(\partial)$.
Then by Theorem \ref{thm:main-adler} $\{\cdot\,_\lambda\,\cdot\}$ is a PVA $\lambda$-bracket
on $\mc V_1$, and \eqref{eq:hierarchy}
uniquely extends to a hierarchy of compatible Hamiltonian equations
over the PVA $\mc V_1$.
However, if the $\lambda$-bracket on $\mc V$
is a PVA $\lambda$-bracket,
then we have a hierarchy of compatible Hamiltonian equations
with Hamiltonian functionals $\tint h_{n,B}$, $n\in\mb Z$, $B$ a $K$-th root of $A$,
on the whole $\mc V$.
\end{remark}
\begin{remark}\label{rem:referee}
The fact that the Hamiltonian functionals $\tint h_{n,B}$, $n\in\mb Z$, $B$ a root of $A$, 
span an infinite dimensional space, should be treated separately,
for example as in \cite{DSKV13}.
\end{remark}
In the remainder of the section we will give a proof of Theorem \ref{thm:hn}.
It is based on the following:
\begin{lemma}\label{lem:hn2}
For $a\in\mc V$, $n\in\mb Z$, and $B(\partial)\in\Mat_{N\times N}\mc V((\partial^{-1}))$
a $K$-th root of $A(\partial)$ ($K\in\mb Z\backslash\{0\}$), we have
\begin{equation}\label{eq:hn2}
\begin{split}
& \{{h_{n,B}}_\lambda a\}\big|_{\lambda=0}
=
-\sum_{i,j=1}^N
\Res_z 
\{A_{ij}(z+x)_x a\}\big(\big|_{x=\partial}(B^{n-K})_{ji}(z)\big)
\,,
\\
& \tint \{a_\lambda h_{n,B}\}\big|_{\lambda=0}
=
- \sum_{i,j=1}^N
\int \Res_z 
\{a_\lambda A_{ij}(z+x)\}\big|_{\lambda=0}\big(\big|_{x=\partial}(B^{n-K})_{ji}(z)\big)
\,.
\end{split}
\end{equation}
\end{lemma}
\begin{proof}
The proof of the first equation in \eqref{eq:hn2}
is similar to the proof of \cite[Lem.3.2]{DSKV15a}.
We review here the argument.
By the second equation in \eqref{eq:lambda-power}, we have, for $n\geq0$,
\begin{equation}\label{eq:hn2-pr1}
\begin{split}
& \{(B^n)_{ij}\!(z)_\lambda a\}
\\
& =
\sum_{\ell=0}^{n-1}\sum_{h,k=1}^N
\{B_{hk}(z+x)_{\lambda+x+y}a\}
\big(\big|_{x=\partial} (B^{n-1-\ell})_{kj}(z)\big)
\big(\big|_{y=\partial} (B^{*\ell})_{hi}(\lambda-z)\big)
\,,
\end{split}
\end{equation}
while,
by the second equation in \eqref{eq:lambda-neg-power}, we have, for $n\leq-1$,
\begin{equation}\label{eq:hn2-pr1b}
\begin{split}
& \{(B^n)_{ij}\!(z)_\lambda a\}
\\
& =
-\sum_{\ell=n}^{-1}\!\sum_{h,k=1}^N
\{B_{hk}(z+x)_{\lambda+x+y}a\}
\big(\big|_{x=\partial} (B^{n-1-\ell})_{kj}(z)\big)
\big(\big|_{y=\partial} (B^{*\ell})_{hi}(\lambda-z)\big)
\,.
\end{split}
\end{equation}
For $K\geq1$, we take $n=K$ in \eqref{eq:hn2-pr1},
replace $z$ by $z+\partial$ and $\lambda$ by $\partial$ acting on $B^{n-K}_{ji}(z)$,
take residue in $z$ and sum over $i,j=1,\dots,N$. 
As a result, we get
\begin{equation}\label{eq:hn2-pr2a}
\begin{split}
& \sum_{i,j=1}^N
\Res_z
\{A_{ij}\!(z+\lambda)_\lambda a\}
\big(\big|_{\lambda=\partial}(B^{n-K})_{ji}(z)\big)
\\
& =
\sum_{\ell=0}^{K-1}\sum_{i,j,h,k=1}^N
\Res_z
\{B_{hk}(z\!+\!\lambda\!+\!x)_{\lambda+x+y}a\}
\\
& \,\,\,\,\,\,\,\,\, \times
\big(\big|_{x=\partial} (B^{K-1-\ell})_{kj}(z+\lambda)\big)
\big(\big|_{y=\partial} (B^{*\ell})_{hi}(-z)\big)
\big(\big|_{\lambda=\partial}B^{n-K}_{ji}(z)\big)
\\
& =
\sum_{\ell=0}^{K-1}\!\sum_{i,j,h,k=1}^N
\Res_z
\{B_{hk}(z+\lambda+x+y)_{\lambda+x+y}a\}
\\
& \,\,\,\,\,\,\,\,\, \times
\big(\big|_{x=\partial} (B^{K-1-\ell})_{kj}(z+\lambda+y)\big)
\big(\big|_{\lambda=\partial}B^{n-K}_{ji}(z+y)\big)
\big(\big|_{y=\partial} (B^{\ell})_{ih}(z)\big)
\\
& =
K \sum_{h,k=1}^N
\Res_z
\{B_{hk}(z+x)_{x}a\}
\big(\big|_{x=\partial} (B^{n-1})_{kh}(z)\big)
\,.
\end{split}
\end{equation}
In the second equality we used Lemma \ref{lem:hn1}(a).
For $K\leq-1$, we do the same manipulation starting 
from equation \eqref{eq:hn2-pr1b} in place of \eqref{eq:hn2-pr1}.
As a result we get the same final equation as \eqref{eq:hn2-pr2a}
(which thus holds for every $K\in\mb Z\backslash\{0\}$):
\begin{equation}\label{eq:hn2-pr2}
\begin{split}
& \sum_{i,j=1}^N
\Res_z
\{A_{ij}\!(z+\lambda)_\lambda a\}
\big(\big|_{\lambda=\partial}(B^{n-K})_{ji}(z)\big)
\\
& =
K \sum_{h,k=1}^N
\Res_z
\{B_{hk}(z+x)_{x}a\}
\big(\big|_{x=\partial} (B^{n-1})_{kh}(z)\big)
\,.
\end{split}
\end{equation}
On the other hand, for $n\geq1$ we put $\lambda=0$ in \eqref{eq:hn2-pr1},
take the residue in $z$, 
let $i=j$ and sum over $i=1,\dots,N$, and as a result we get, by the definition \eqref{eq:hn} of $h_{n,B}$,
\begin{equation}\label{eq:hn2-pr3}
\begin{split}
& \{{h_{n,B}}_\lambda a\}\big|_{\lambda=0}
=
\frac{-K}{n} \sum_{i=1}^N \Res_z \{(B^n)_{ii}(z)_\lambda a\}\big|_{\lambda=0}
\\
& =
\frac{-K}{n} 
\sum_{\ell=0}^{n-1}\sum_{i,h,k=1}^N
\Res_z \{B_{hk}(z+x)_{x+y}a\}
\big(\big|_{x=\partial} (B^{n-1-\ell})_{ki}(z)\big)
\big(\big|_{y=\partial} (B^{*\ell})_{hi}(-z)\big)
\\
& =
\!\frac{-K}{n}
\!\sum_{\ell=0}^{n-1}\!\sum_{i,h,k=1}^N
\Res_z \{B_{hk}(z\!+\!x\!+\!y)_{x+y}a\}
\big(\big|_{x=\partial} (B^{n-1-\ell})_{ki}(z+y)\big)
\big(\big|_{y=\partial} (B^{\ell})_{ih}(z)\big)
\\
& =
-K \sum_{h,k=1}^N
\Res_z \{B_{hk}(z+x)_{x}a\}
\big(\big|_{x=\partial} (B^{n-1})_{kh}(z)\big)
\,.
\end{split}
\end{equation}
Again, in the third equality we used Lemma \ref{lem:hn1}(a).
We arrive at the same conclusion, for $n\leq-1$, using equation \eqref{eq:hn2-pr1b}
in place of \eqref{eq:hn2-pr1}.
Combining equations \eqref{eq:hn2-pr2} and \eqref{eq:hn2-pr3}, we get 
the first equation in \eqref{eq:hn2}.

Next, we prove the second equation in \eqref{eq:hn2}. 
By the first equation in \eqref{eq:lambda-power}, we have, for $n\geq1$,
\begin{equation}\label{eq:hn2-pr4}
\{a_\lambda (B^n)_{ij}(z)\}
=
\sum_{\ell=0}^{n-1}\sum_{h,k=1}^N
(B^{n-1-\ell})_{ih}(z+\lambda+\partial)
\{a_\lambda B_{hk}(z+x)\}
\big(\big|_{x=\partial}(B^\ell)_{kj}(z)\big)
\,,
\end{equation}
while by the first equation in \eqref{eq:lambda-neg-power}, we have, 
for $n\leq-1$,
\begin{equation}\label{eq:hn2-pr4b}
\{a_\lambda (B^n)_{ij}(z)\}
=
-\sum_{\ell=n}^{-1}\sum_{h,k=1}^N
(B^{n-1-\ell})_{ih}(z+\lambda+\partial)
\{a_\lambda B_{hk}(z+x)\}
\big(\big|_{x=\partial}(B^\ell)_{kj}(z)\big)
\,.
\end{equation}
For $K\geq1$, 
we let $n=K$ in \eqref{eq:hn2-pr4},
let $\lambda=0$, replace $z$ by $z+\partial$ acting on $B^{n-K}_{ji}(z)$,
take residue in $z$, apply $\tint$, and sum over $i,j=1,\dots,N$.
As a result we get
\begin{equation}\label{eq:hn2-pr5a}
\begin{split}
& \sum_{i,j=1}^N
\int\Res_z
\{a_\lambda (B^K)_{ij}(z+x)\}\big|_{\lambda=0} \big(\big|_{x=\partial}B^{n-K}_{ji}(z)\big)
\\
& =
\sum_{\ell=0}^{K-1}\sum_{i,h,k=1}^N
\int\Res_z
(B^{K-1-\ell})_{ih}(z+\partial)
\{a_\lambda B_{hk}(z+x)\}\big|_{\lambda=0}
\big(\big|_{x=\partial}(B^{n+\ell-K})_{ki}(z)\big)
\\
& =
\sum_{\ell=0}^{K-1}\sum_{i,h,k=1}^N
\int\Res_z
(B^{n+\ell-K})_{ki}(z+\partial)
(B^{K-1-\ell})_{ih}(z+\partial)
\{a_\lambda B_{hk}(z)\}\big|_{\lambda=0}
\\
& =
K\sum_{h,k=1}^N
\int\Res_z
(B^{n-1})_{kh}(z+\partial)
\{a_\lambda B_{hk}(z)\}\big|_{\lambda=0}
\,.
\end{split}
\end{equation}
In the second equality we used Lemma \ref{lem:hn1}(b).
For $K\leq-1$, we do the same manipulations with equation \eqref{eq:hn2-pr4b}
in place of \eqref{eq:hn2-pr4}
and, as a result, we get the same final equation as \eqref{eq:hn2-pr5a}
(which thus holds for evey $K\in\mb Z\backslash\{0\}$):
\begin{equation}\label{eq:hn2-pr5}
\begin{split}
& \sum_{i,j=1}^N
\int\Res_z
\{a_\lambda (B^K)_{ij}(z+x)\}\big|_{\lambda=0} \big(\big|_{x=\partial}B^{n-K}_{ji}(z)\big)
\\
& =
K\sum_{h,k=1}^N
\int\Res_z
(B^{n-1})_{kh}(z+\partial)
\{a_\lambda B_{hk}(z)\}\big|_{\lambda=0}
\,.
\end{split}
\end{equation}
On the other hand, for $n\geq1$ we put $\lambda=0$ in \eqref{eq:hn2-pr4},
take the residue in $z$, apply $\tint$, 
let $i=j$ and sum over $i=1,\dots,N$.
As a result we get, by the definition \eqref{eq:hn} of $h_{n,B}$,
\begin{equation}\label{eq:hn2-pr6}
\begin{split}
& \tint\{a_\lambda h_{n,B}\}\big|_{\lambda=0}
= 
\frac{-K}{n} 
\sum_{i=1}^N
\int\Res_z
\{a_\lambda (B^n)_{ii}(z)\}\big|_{\lambda=0}
\\
& =
\frac{-K}{n} 
\sum_{\ell=0}^{n-1}\sum_{i,h,k=1}^N
\int\Res_z
(B^{n-1-\ell})_{ih}(z+\partial)
\{a_\lambda B_{hk}(z+x)\}\big|_{\lambda=0}
\big(\big|_{x=\partial}(B^\ell)_{ki}(z)\big)
\\
& =
\frac{-K}{n} 
\sum_{\ell=0}^{n-1}\sum_{i,h,k=1}^N
\int\Res_z
(B^\ell)_{ki}(z+\partial)
(B^{n-1-\ell})_{ih}(z+\partial)
\{a_\lambda B_{hk}(z)\}\big|_{\lambda=0}
\\
& =
-K 
\sum_{h,k=1}^N
\int\Res_z
(B^{n-1})_{kh}(z+\partial)
\{a_\lambda B_{hk}(z)\}\big|_{\lambda=0}
\,.
\end{split}
\end{equation}
Again, in the third equality we used Lemma \ref{lem:hn1}(b).
For $n\leq-1$ we arrive at the same conclusion, using equation \eqref{eq:hn2-pr4b}
in place of \eqref{eq:hn2-pr4}.
Combining equations \eqref{eq:hn2-pr5} and \eqref{eq:hn2-pr6}, we get 
the second equation in \eqref{eq:hn2}.
\end{proof}
\begin{proof}[Proof of Theorem \ref{thm:hn}]
Suppose $B$ is a $K$-th root of $A$, $K\in\mb Z\backslash\{0\}$ 
and $C$ is an $L$-th root of $A$, $L\in\mb Z\backslash\{0\}$.
Applying the second equation in \eqref{eq:hn2} first,
and then the first equation in \eqref{eq:hn2}, we get
\begin{equation}\label{eq:hn-pr1}
\begin{split}
& \{\tint h_{m,B},\tint h_{n,C}\}
= 
\sum_{i,j,h,k=1}^N
\int \Res_z \Res_w 
\{A_{ij}(z+x)_x A_{hk}(w+y)\}
\\
& \,\,\,\,\,\,\,\,\,\,\,\,\,\,\,\,\,\, \times
\big(\big|_{x=\partial}(B^{m-K})_{ji}(z)\big)
\big(\big|_{y=\partial}(C^{n-L})_{kh}(w)\big)
\,.
\end{split}
\end{equation}
We can now use the Adler condition \eqref{eq:adler}
to rewrite the RHS of \eqref{eq:hn-pr1} as
\begin{equation}\label{eq:hn-pr2}
\begin{split}
& \sum_{i,j,h,k=1}^N
\!\!
\int\! \Res_z\! \Res_w\! 
A_{hj}(w\!+\!\partial)
(B^{m-K})_{ji}\!(z)
\iota_z(z\!-\!w\!-\!\partial)^{-1}
(A_{ik})^*(-\!z)
(C^{n-L})_{kh}\!(w)
\\
& -
\sum_{i,h=1}^N
\int \Res_z \Res_w 
(B^{m})_{hi}(z)
\iota_z(z\!-\!w\!-\!\partial)^{-1}
(C^{n})_{ih}(w)
\,.
\end{split}
\end{equation}
For the second term we used the fact that $A=B^K=C^L$.
By \eqref{eq:positive}, we have
\begin{equation}\label{eq:hn-pr2b}
\sum_{i=1}^N
\Res_z
(B^{m-K})_{ji}(z)
\iota_z(z-w-\partial)^{-1}
\circ
(A_{ik})^*(-z)
=
(B^m)_{jk}(w+\partial)_+
\,.
\end{equation}
Hence, the first term of \eqref{eq:hn-pr2} is equal to
\begin{equation}\label{eq:hn-pr3}
\begin{split}
& \sum_{j,h,k=1}^N
\int \Res_w 
A_{hj}(w+\partial)
(B^m)_{jk}(w+\partial)_+
(C^{n-L})_{kh}(w)
\\
& =
\sum_{j,k=1}^N
\int \Res_w 
(C^{n})_{kj}(w+\partial)
(B^m)_{jk}(w)_+
\,.
\end{split}
\end{equation}
Here we used Lemma \ref{lem:hn1}(b) and the identity $C^L=A$.
On the other hand,
by \eqref{eq:positive}
the second term of \eqref{eq:hn-pr2} is equal to
\begin{equation}\label{eq:hn-pr4}
\begin{split}
& -
\sum_{i,h=1}^N
\int \Res_w 
(B^{m})_{hi}(w+\partial)_+
(C^{n})_{ih}(w)
\end{split}
\end{equation}
which, by Lemma \ref{lem:hn1}(b), is the same as the RHS of \eqref{eq:hn-pr3}, with opposite sign.
In conclusion, \eqref{eq:hn-pr2} is zero, as claimed.
This proves part (a).

For part (b),
by the first equation in Lemma \ref{eq:hn2} we have
\begin{equation}\label{eq:proofb-1}
\begin{split}
& \{\tint h_{n,B},A_{hk}(w)\}
=
\{{h_{n,B}}_\lambda A_{hk}(w)\}\big|_{\lambda=0}
\\
& =
- \sum_{i,j=1}^N
\Res_z 
\{A_{ij}(z+x)_x A_{hk}(w)\}\big(\big|_{x=\partial}(B^{n-K})_{ji}(z)\big)
\,.
\end{split}
\end{equation}
By the Adler condition \eqref{eq:adler}
the RHS of \eqref{eq:proofb-1} becomes
\begin{equation}\label{eq:proofb-2}
\begin{split}
& -\sum_{i,j=1}^N
\Res_z 
A_{hj}(w+\partial)
(B^{n-K})_{ji}(z)
\iota_z(z\!-\!w\!-\!\partial)^{-1}(A_{ik})^*(-z)
\\
&\,\,\,\,\,\, +
\sum_{i,j=1}^N
\Res_z 
\big(
A_{hj}(z+\partial)
(B^{n-K})_{ji}(z)\big)
\iota_z(z\!-\!w\!-\!\partial)^{-1}A_{ik}(w)
\\
& =
- \sum_{j=1}^N
\Res_z 
A_{hj}(w+\partial)
(B^{n})_{jk}(z)
\iota_z(z\!-\!w)^{-1}
\\
&\,\,\,\,\,\, +
\sum_{i=1}^N
\Res_z 
(B^{n})_{hi}(z)
\iota_z(z\!-\!w\!-\!\partial)^{-1}A_{ik}(w)
\\
& =
\sum_{i=1}^N
\Big(
-
A_{hi}(w+\partial)
(B^{n})_{ik}(w)_+
+ 
(B^{n})_{hi}(w+\partial)_+ A_{ik}(w)
\Big)\,.
\end{split}
\end{equation}
In the second equality we used Lemma \ref{lem:hn1} and the identity $A=B^K$,
while in the last equality we used \eqref{eq:positive}.
This proves \eqref{eq:hierarchy} and completes the proof of the Theorem.
\end{proof}

\section{Operators of bi-Adler type and bi-Hamiltonian hierarchies}\label{sec:6}

\subsection{Operators of bi-Adler type and bi-Poisson vertex algebras}\label{sec:6.1}

Let $\{\cdot\,_\lambda\,\cdot\}_0$ and $\{\cdot\,_\lambda\,\cdot\}_1$
be two $\lambda$-brackets on the same differential algebra $\mc V$.
We can consider the pencil of $\lambda$-brackets
\begin{equation}\label{eq:pencil}
\{\cdot\,_\lambda\,\cdot\}_\epsilon
=
\{\cdot\,_\lambda\,\cdot\}_0
+
\epsilon\{\cdot\,_\lambda\,\cdot\}_1
\quad,\qquad \epsilon\in\mb F
\,.
\end{equation}
We say that $\mc V$ is a \emph{bi}-\emph{PVA} 
if $\{\cdot\,_\lambda\,\cdot\}_\epsilon$ is a PVA $\lambda$-bracket on $\mc V$ for every $\epsilon\in\mb F$.
Clearly, for this it suffices that $\{\cdot\,_\lambda\,\cdot\}_0$,
$\{\cdot\,_\lambda\,\cdot\}_1$ and $\{\cdot\,_\lambda\,\cdot\}_0+\{\cdot\,_\lambda\,\cdot\}_1$
are PVA $\lambda$-brackets.

\begin{definition}\label{def:-bi-adler}
Let $A(\partial)\in\Mat_{M\times N}\mc V((\partial^{-1}))$
be a matrix pseudodifferential operator,
and let $S\in\Mat_{M\times N}\mb F$.
We say that $A$ is of $S$-\emph{Adler type}
with respect to the two $\lambda$-brackets $\{\cdot\,_\lambda\,\cdot\}_0$
and $\{\cdot\,_\lambda\,\cdot\}_1$
if, for every $\epsilon\in\mb F$,
$A(\partial)+\epsilon S$ is a matrix of Adler type with respect 
to the $\lambda$-bracket $\{\cdot\,_\lambda\,\cdot\}_\epsilon$.
This is equivalent to saying that $A(\partial)$
is a matrix of Adler type with respect to the $\lambda$-bracket $\{\cdot\,_\lambda\,\cdot\}_0$,
i.e. \eqref{eq:adler} holds,
and that 
\begin{equation}\label{eq:B-adler}
\begin{split}
\{A_{ij}(z)_\lambda A_{hk}(w)\}_1
& = 
\iota_z(z\!-\!w\!-\!\lambda)^{-1}
S_{ik}
\big(
A_{hj}(w+\lambda)
- A_{hj}(z)
\big)
\\
& +
\iota_z(z\!-\!w\!-\!\lambda\!-\!\partial)^{-1}
S_{hj}
\big(
(A_{ik})^*(\lambda-z)
-A_{ik}(w)
\big)
\,.
\end{split}
\end{equation}
In the case $M=N$,
we also say that $A$ is of \emph{bi}-\emph{Adler type}
if it is of $\id_N$-Adler type.
In this case condition \eqref{eq:B-adler} becomes:
\begin{equation}\label{eq:bi-adler}
\begin{split}
\{A_{ij}(z)_\lambda A_{hk}(w)\}_1
& = 
\delta_{ik}
\iota_z(z\!-\!w\!-\!\lambda)^{-1}
\big(
A_{hj}(w+\lambda)
- A_{hj}(z)
\big)
\\
& +
\delta_{hj}
\iota_z(z\!-\!w\!-\!\lambda\!-\!\partial)^{-1}
\big(
(A_{ik})^*(\lambda-z)
-A_{ik}(w)
\big)
\,.
\end{split}
\end{equation}
\end{definition}
\begin{example}\label{ex:bi-adler-glN}
From Example \ref{ex:affine-adler} we have that,
for every $S\in\Mat_{N\times N}\mb F$,
the matrix differential operator
$A(\partial)=\id_N\partial+\sum_{i,j=1}^Nq_{ji}E_{ij}$
is of $S$-Adler type with respect the the bi-PVA $\lambda$-brackets 
(cf. \eqref{eq:affine-lb})
\begin{equation}\label{eq:affine-bi-lb}
\{a_\lambda b\}_0=[a,b]+\tr(ab)\lambda
\quad,\qquad
\{a_\lambda b\}_1=\tr(S[a,b])
\,\,,\,\,\,\,
a,b\in\mf{gl}_N
\,,
\end{equation}
extended to $\mc V(\mf{gl}_N)$ by sesquilinearity and the Leibniz rules.
\end{example}

Operators of $S$-Adler type are related to bi-PVA's
in the same way as operators of Adler type are related to PVA's.
In other words, we have the following extension of Theorem \ref{thm:main-adler},
which, in fact, is a consequence of Theorem \ref{thm:main-adler}:
\begin{theorem}\label{thm:main-bi-adler}
Let $S\in\Mat_{M\times N}\mb F$,
and let $A(\partial)\in\Mat_{M\times N}\mc V((\partial^{-1}))$ be an $M\times N$-matrix 
pseudodifferential operator of $S$-Adler type with respect to the $\lambda$-brackets 
$\{\cdot\,_\lambda\,\cdot\}_0$ and $\{\cdot\,_\lambda\,\cdot\}_1$ on $\mc V$.
Assume that the coefficients of the entries of the matrix $A(\partial)$
generate $\mc V$ as a differential algebra.
Then $\mc V$ is a bi-PVA with the $\lambda$-brackets $\{\cdot\,_\lambda\,\cdot\}_0$
and $\{\cdot\,_\lambda\,\cdot\}_1$.
\end{theorem}

\subsection{Bi-Hamiltonian hierarchy associated to an operator of bi-Adler type}\label{sec:6.2}

Let $\mc V$ be a bi-Poisson vertex algebra with $\lambda$-brackets 
$\{\cdot\,_\lambda\,\cdot\}_0$ and $\{\cdot\,_\lambda\,\cdot\}_1$.
A \emph{bi}-\emph{Hamiltonian equation} is an evolution equation 
which can be written in Hamiltonian form with respect to both PVA $\lambda$-brackets
and two Hamiltonian functionals $\tint h_0,\tint h_1\in\mc V/\partial\mc V$:
$$
\frac{du}{dt}
=
\{\tint h_0,u\}_0
=
\{\tint h_1,u\}_1
\,,\,\, u\in\mc V
\,.
$$
The usual way to prove integrability for a bi-Hamiltonian equation
is to solve the so called Lenard-Magri
recurrence relation ($u\in\mc V$):
\begin{equation}\label{eq:LM}
\{\tint h_n,u\}_0
=
\{\tint h_{n+1},u\}_1
\,\,,\,\,\,\,
n\in\mb Z_+
\,.
\end{equation}
In this way, 
we get the corresponding hierarchy of bi-Hamiltonian equations
$$
\frac{du}{dt_n}
=
\{\tint h_n,u\}_0
=
\{\tint h_{n+1},u\}_1
\,,\,\, 
n\in\mb Z_+,\, u\in\mc V
\,.
$$

We next show that,
given an operator 
$A(\partial)\in\Mat_{N\times N}\mc V((\partial^{-1}))$ of bi-Adler type,
the sequence of local functionals provided by Theorem \ref{thm:hn}
satisfies a generalized Lenard-Magri recurrence relation,
and therefore gives a hierarchy of bi-Hamiltonian equations.
\begin{theorem}\label{thm:bi-hn}
Let $\mc V$ be a differential algebra with two $\lambda$-brackets 
$\{\cdot\,_\lambda\,\cdot\}_0$ and $\{\cdot\,_\lambda\,\cdot\}_1$.
Let $A(\partial)\in\Mat_{N\times N}\mc V((\partial^{-1}))$
be a matrix pseudodifferential operator of bi-Adler type 
with respect to the $\lambda$-brackets $\{\cdot\,_\lambda\,\cdot\}_0$ 
and $\{\cdot\,_\lambda\,\cdot\}_1$,
and assume that $A(\partial)$ is invertible in $\Mat_{N\times N}\mc V((\partial^{-1}))$.
Let $B(\partial)\in\Mat_{N\times N}\mc V((\partial^{-1}))$
be a $K$-th root of $A$.
Then, 
the elements $h_{n,B}\in\mc V$, $n\in\mb Z_+$, given by \eqref{eq:hn}
satisfy the following generalized Lenard-Magri recurrence relation:
\begin{equation}\label{eq:LM-K}
\{\tint h_{n,B},A(z)\}_0
=
\{\tint h_{n+K,B},A(z)\}_1
=
[(B^n)_+,A](z)
\,\,,\,\,\,\,
n\in\mb Z
\,.
\end{equation}
Hence, \eqref{eq:hierarchy} is a compatible hierarchy of bi-Hamiltonian equations
on the bi-PVA $\mc V_1$ generated by the coefficients of the entries of $A(z)$.
Moreover, all the Hamiltonian functionals $\tint h_{n,C}$, $n\in\mb Z_+$, $C$ root of $A$,
are integrals of motion of all the equations of this hierarchy.
\end{theorem}
\begin{remark}\label{rem:negative-n-bi}
As in Remark \ref{rem:negative-n},
one can state the same Theorem \ref{thm:bi-hn} without the assumption that $A$ is invertible.
In this case we need to restrict to $K\geq1$,
and the proof of the Lenard-Magri recurrence relation \eqref{eq:LM-K}
only works for $n\geq K$.
\end{remark}
\begin{remark}\label{rem:V1-bi}
By Theorem \ref{thm:main-adler} the $\lambda$-brackets
$\{\cdot\,_\lambda\,\cdot\}_0$ and $\{\cdot\,_\lambda\,\cdot\}_1$
define a bi-PVA structure on $\mc V_1$.
The Lenard-Magri recurrence
$\{\tint h_{n,B},u\}_0=\{\tint h_{n+K,B},u\}_1$
holds for every $n\in\mb Z_+$ and $u\in\mc V_1$,
and \eqref{eq:hierarchy} uniquely extends to a hierarchy of compatible bi-Hamiltonian equations
over the bi-PVA $\mc V_1$.
However, the Lenard-Magri recurrence equation does not necessarily hold for every $u\in\mc V$
even if the two $\lambda$-brackets $\{\cdot\,_\lambda\,\cdot\}_0$ and $\{\cdot\,_\lambda\,\cdot\}_1$
of $\mc V$ are compatible PVA $\lambda$-brackets.
This phenomenon does actually occur even in the simplest example discussed in Section \ref{sec:7.1a}.
\end{remark}
\begin{remark}\label{rem:victor}
Suppose that $\mc V$ is a differential field
and $A(\partial)\in\Mat_{N\times N}\mc V((\partial^{-1}))$
is an operator of bi-Adler type
with respect to two $\lambda$-brackets $\{\cdot\,_\lambda\,\}_0$ and $\{\cdot\,_\lambda\,\}_1$.
Then $A(\partial)+\epsilon\id_N$ is again an operator of bi-Adler type,
with respect to the $\lambda$-brackets $\{\cdot\,_\lambda\,\}_0+\epsilon\{\cdot\,_\lambda\,\}_1$ 
and $\{\cdot\,_\lambda\,\}_1$.
Note that $A(\partial)+\epsilon\id_N$ is invertible for all but finitely many
values of $\epsilon\in\mb F$.
Hence, in Theorem \ref{thm:bi-hn} we do not need to require that $A(\partial)$
is invertible.
\end{remark}
\begin{proof}[Proof of Theorem \ref{thm:bi-hn}]
By Theorem \ref{thm:hn}(b), we have
\begin{equation}\label{eq:bi-proof1}
\{\tint h_{n,B},A_{hk}(w)\}_0
=
[(B^n)_+,A]_{hk}(w)
\,,\,\,n\in\mb Z_+\,.
\end{equation}
On the other hand, by the first equation in \eqref{eq:hn2} we have
\begin{equation}\label{eq:bi-proof2}
\begin{split}
& \{\tint{h_{n,B}},A_{hk}(w)\}_1
=
\{{h_{n,B}}_\lambda A_{hk}(w)\}_1\big|_{\lambda=0}
\\
& =
-\sum_{i,j=1}^N
\Res_z 
\{A_{ij}(z+x)_x A_{hk}(w)\}_1\big(\big|_{x=\partial}(B^{n-K})_{ji}(z)\big)
\,.
\end{split}
\end{equation}
Applying the bi-Adler identity \eqref{eq:bi-adler}
we can rewrite the RHS of \eqref{eq:bi-proof2} as
\begin{equation}\label{eq:bi-proof3}
\begin{split}
& -\sum_{j=1}^N
\Res_z 
\iota_z(z\!-\!w)^{-1}
\big(
A_{hj}(w+\partial)
- A_{hj}(z+\partial)
\big)
(B^{n-K})_{jk}(z)
\\
& -
\sum_{i=1}^N
\Res_z 
(B^{n-K})_{hi}(z)
\iota_z(z\!-\!w\!-\!\partial)^{-1}
\big(
(A_{ik})^*(-z)
-A_{ik}(w)
\big)
\\
&= -\sum_{j=1}^N
A_{hj}(w+\partial)
(B^{n-K})_{jk}(w)
+\sum_{i=1}^N
(B^{n-K})_{hi}(w+\partial)_+
A_{ik}(w)
\\
&= [(B^{n-K})_+,A]_{hk}(w)
\,.
\end{split}
\end{equation}
In the first equality we used equations \eqref{eq:positive}
and \eqref{eq:hn-pr2b} and the identities $A\circ B^{n-K}=B^n=B^{n-K}\circ A$.
Comparing \eqref{eq:bi-proof1} and \eqref{eq:bi-proof3},
we get \eqref{eq:LM-K}.
The last statement of the theorem is an immediate consequence of Theorem \ref{thm:hn}
and \eqref{eq:LM-K}.
\end{proof}

\subsection{A new scheme of integrability of bi-Hamiltonian PDE}\label{sec:6.3}

Based on Theorems \ref{thm:main-quasidet}, \ref{thm:hn} and \ref{thm:bi-hn},
we propose the following scheme of integrability.
Let $\mc V$ be a differential algebra with compatible PVA $\lambda$-brackets 
$\{\cdot\,_\lambda\,\cdot\}_0$ and $\{\cdot\,_\lambda\,\cdot\}_1$.
Let $S\in\Mat_{N\times N}\mb F$ 
and let $A(\partial)\in\Mat_{N\times N}\mc V((\partial^{-1}))$
be an operator of  $S$-Adler type with respect to the
$\lambda$-brackets $\{\cdot\,_\lambda\,\cdot\}_0$ and $\{\cdot\,_\lambda\,\cdot\}_1$.
Assume (without loss of generality) that the differential algebra $\mc V$
is generated by the coefficients of $A(\partial)$.
Then, we obtain an integrable hierarchy of bi-Hamiltonian equations
as follows:
\begin{enumerate}[1.]
\item
consider the canonical factorization $S=IJ$,
where $J:\,\mb F^N\twoheadrightarrow\im(S)$
and $I:\,\im(S)\hookrightarrow\mb F^N$;
\item
assume that the $(I,J)$-quasideterminant $|A|_{IJ}(\partial)$ exists;
then, by Theorem \ref{thm:main-quasidet} and Proposition \ref{thm:quasidet-adler}
$|A|_{IJ}$ is an $M\times M$ matrix pseudodifferential operator
(where $M=\dim\im(S)$) of bi-Adler type
with respect to the same $\lambda$-brackets 
$\{\cdot\,_\lambda\,\cdot\}_0$ and $\{\cdot\,_\lambda\,\cdot\}_1$;
\item
consider the family of local functionals
$\{\tint h_{n,B}\,|\,n\in\mb Z,\,B \text{ a } K\text{-th root of } |A|_{IJ}\}$
given by \eqref{eq:hn};
then, by Theorem \ref{thm:bi-hn},
they are all Hamiltonian functionals in involution with respect to both PVA $\lambda$-brackets
$\{\cdot\,_\lambda\,\cdot\}_0$ and $\{\cdot\,_\lambda\,\cdot\}_1$,
and they satisfy the Lenard-Magri recurrence relation \eqref{eq:LM-K};
\item
we thus get an integrable hierarchy of bi-Hamiltonian equations
\begin{equation}\label{eq:bi-hierarchy}
\frac{du}{dt_{n,B}}
=
\{\tint h_{n,B},u\}_0
=
\{\tint h_{n+K,B},u\}_1
\,,
\end{equation}
provided that the $\tint h_{n,B}$ span an infinite dimensional space.
\end{enumerate}
In Section \ref{sec:7}
we will show how to implement this scheme to the affine 
bi-PVA's for $\mf{gl}_N$ defined in Example \ref{ex:affine-pva} 
and the corresponding operators of bi-Adler type in Example \ref{ex:affine-adler}.

In forthcoming publications we plan to implement the same scheme
to construct integrable hierarchies associated to the classical $\mc W$-algebras
$\mc W(\mf g,f)$ for classical Lie algebras $\mf g$ 
and arbitrary nilpotent elements $f\in\mf g$.

\begin{remark}\label{rem:casimir}
Consider the space $C_1$ of Casimir functionals with respect 
to the $\lambda$-bracket $\{\cdot\,_\lambda\,\cdot\}_1$:
\begin{equation}\label{eq:casimir}
C_1=
\big\{\tint h\in\mc V/\partial\mc V\,\big|\, \{\tint h,\mc V\}_1=0 \big\}
\,.
\end{equation}
Note that, if $A(\partial)$ is an operator of bi-Adler type,
then $\tint h_{1,A}=\tint\Res_z\tr A(z)$ is a Casimir functional
(this follows by Theorem \ref{thm:bi-hn}, or it can be easily checked directly).
By Theorem \ref{thm:bi-hn},
for this Casimir element the Lenard-Magri recurrence scheme \eqref{eq:LM}
is infinite, namely it admits an infinite sequence $\{\tint h_n\}_{n\in\mb Z_+}$
of solutions, starting from $\tint h_0=\tint h_{1,A}$.
We conjecture that,
if $A(\partial)$ is a matrix pseudodifferential operator of bi-Adler type,
then the Lenard-Magri scheme \eqref{eq:LM} is infinite
for each Casimir functional $\tint h_0\in C_1$.
Note that, by \cite[Prop.2.10(c)]{BDSK09} all the resulting Hamiltonian functionals are in involution.
\end{remark}
\begin{remark}\label{rem:casimir2}
Sometimes $A(\partial)$ is of bi-Adler type with respect to a certain bi-PVA $\mc V$,
and the differential subalgebra $\mc V_1\subset\mc V$ generated by its coefficients
is a proper bi-PVA subalgebra.
In such case it may happen, as the example discussed in Section \ref{sec:7.1a} shows,
that $\tint h_{1,A}=\tint\Res_z\tr A(z)$ is a Casimir functional
with respect to the $\lambda$-bracket $\{\cdot\,_\lambda\,\cdot\}_1$ on $\mc V_1$
but not on $\mc V$.
\end{remark}

\section{
Example: bi-Hamiltonian hierarchies for the affine PVA
\texorpdfstring{$\mc V_S(\mf{gl}_2)$}{V\_S(gl\_2)}}
\label{sec:7}

Consider the affine PVA $\mc V=\mc V_S(\mf{gl}_2)$ from Example \ref{ex:affine-pva}
and the matrix differential operator of Adler type
$A(\partial)+S^t$, where $A(\partial)=\partial\id_2+
\left(\begin{array}{ll}
q_{11}&q_{21}\\
q_{12}&q_{22}
\end{array}\right)$, from Example \ref{ex:affine-adler}.
We shall study the bi-Hamiltonian hierarchies associated to this operator of $S$-Adler type,
for various choices of $S\in\Mat_{2\times2}\mb F$,
following the general method discussed in Section \ref{sec:6}.

Note that if $S$ is non-degenerate,
then its canonical decomposition is $S=S\id_2$
and the corresponding $(S\id_2)$-quasideterminant of $A(\partial)$ is
$|A(\partial)|_{S\id_2}=S^{-1}A(\partial)$, which is a differential operator of order $1$.
Hence, the residue of any power of it is zero.
Also, by the results of Appendix \ref{sec:app},
taking square roots of $A(\partial)$ only seem to lead to trivial hierarchies as well
(cf. Remark \ref{rem:roots1}).

Hence, we only need to consider the case when $S$ has rank $1$.
In Section \ref{sec:7.1a} we consider the case when $S$ is diagonalizable of rank $1$,
while in Section \ref{sec:7.1b} we consider the case when $S$ is nilpotent of rank $1$.

\subsection{First case: \texorpdfstring{$S=\diag(s,0)$}{S=diag(s,0)}, $s\in\mb F\backslash\{0\}$}
\label{sec:7.1a}

We know that the family of Adler type operators $A(\partial)+\epsilon S$, $\epsilon\in\mb F\backslash\{0\}$,
defines two compatible PVA $\lambda$-brackets 
$\{\cdot\,_\lambda\,\cdot\}_0$ and $\{\cdot\,_\lambda\,\cdot\}_1$ on $\mc V$,
given by (cf. \eqref{eq:affine-lb})
\begin{equation}\label{20150625:eq1}
\{{q_{ij}}_{\lambda}q_{hk}\}_0
=
\delta_{jh}q_{ik}-\delta_{ki}q_{hj}+\delta_{jh}\delta_{ik}\lambda
\,,\,\,
i,j,h,k\in\{1,2\}
\,,
\end{equation}
and the only non-zero $\lambda$-brackets among generators for $\{\cdot\,_\lambda\,\cdot\}_1$ are
\begin{equation}\label{20150625:eq1b}
\{{q_{12}}_{\lambda}q_{21}\}_1=s
\,\,,\,\,\,\,
\{{q_{21}}_{\lambda}q_{12}\}_1=-s
\,.
\end{equation}
Let us consider the canonical factorization $B=IJ$ where
$I=\left(\begin{array}{l}
1\\0
\end{array}\right)$
and
$J=\left(\begin{array}{ll}s&0\end{array}\right)$, and define
\begin{equation}\label{eq:gl2-diag-L}
L(\partial)=|A(\partial)|_{IJ}=\frac1s|A(\partial)|_{11}
=\frac1s\left(\partial+q_{11}-q_{21}(\partial+q_{22})^{-1}q_{12}
\right)
\,.
\end{equation}
Let $\mc V_1\subset\mc V$ be the differential subalgebra generated by the coefficients of $L(\partial)$.
By Thorem \ref{thm:main-bi-adler} it is a bi-PVA subalgebra of $\mc V$.
According to the general method discussed in Section \ref{sec:6.3},
$L(\partial)$ is a pseudodifferential operator of bi-Adler type
with respect to the bi-PVA $\lambda$-brackets \eqref{20150625:eq1} and \eqref{20150625:eq1b}.
Hence, the local functionals $\tint h_0=0$,
$\tint h_n=\frac{-1}{n}\tint\Res_zL^n(z)$, $n\geq1$,
are in involution with respect to both PVA $\lambda$-brackets 
$\{\cdot\,_\lambda\,\cdot\}_0$ and $\{\cdot\,_\lambda\,\cdot\}_1$,
and they satisfy the Lenard-Magri recurrence relation
\begin{equation}\label{eq:LM-gl2}
\{\tint h_n,u\}_0=\{\tint h_{n+1},u\}_1
\,\,,\,\,\,\, n\geq0\,,\,\, u\in\mc V_1
\,.
\end{equation}
We thus get the corresponding integrable hierarchy of Hamiltonian equations
\begin{equation}\label{eq:hier-gl2}
\frac{du}{dt_n}=\{\tint h_n,u\}_0
\,\,,\,\,\,\, n\in\mb Z_+\,,\,\, u\in\mc V
\,,
\end{equation}
which is a bi-Hamiltonian hierarchy when restricted to $u\in\mc V_1$.

For example, from \eqref{eq:gl2-diag-L} we get the first Hamiltonian functional of the sequence:
$$
\tint h_{1}=-\tint \res_zL(z)=\int\frac1sq_{12}q_{21}
\,.
$$
Note that this is a Casimir functional for the first Poisson structure on $\mc V_1$
but not on the whole $\mc V$ (cf. Remark \ref{rem:casimir2}).
By a straightforward computation we get
\begin{equation*}
\begin{split}
L^2(\partial)
&=
\frac1{s^2}\Big(
\partial^2+2q_{11}\partial+q_{11}'+q_{11}^2-2q_{12}q_{21}
-\left(q_{21}'+q_{21}(q_{11}\!-\!q_{22})\right)(\partial+q_{22})^{-1}q_{12}
\\
&-q_{21}(\partial+q_{22})^{-1}\left(-q_{12}'+q_{12}(q_{11}-q_{22})\right)
+\left(q_{21}(\partial+q_{22})^{-1}q_{12}\right)^2
\Big)
\,,
\end{split}
\end{equation*}
and taking its residue we get the second Hamiltonian functional of the sequence:
$$
\tint h_{2}=-\frac12\tint \res_zL(z)=\int\frac1{s^2}
\left(-q_{12}'q_{21}+q_{12}q_{21}(q_{11}-q_{22})\right)
\,.
$$
The corresponding first two equations of the hierarchy are
\begin{equation*}
\begin{split}
&\frac{d q_{12}}{dt_{1}}
=\frac{1}{s}\left(q_{12}^\prime-q_{12}(q_{11}-q_{22})\right)
\,,
\\
&\frac{d q_{21}}{dt_{1}}
=\frac{1}{s}\left(q_{21}^\prime+q_{21}(q_{11}-q_{22})\right)
\,,
\\
&\frac{d q_{12}}{dt_{2}}
=\frac{1}{s}\left(q_{12}^{\prime\prime}
-2q_{12}^\prime(q_{11}-q_{22})-q_{12}(q_{11}^\prime-q_{22}^\prime)+q_{12}(q_{11}-q_{22})^2
-2q_{12}^2q_{21}\right)
\,,
\\
&\frac{d q_{21}}{dt_{2}}
=\frac{1}{s}\left(-q_{21}^{\prime\prime}
-2q_{21}^\prime(q_{11}-q_{22})-q_{21}(q_{11}^\prime-q_{22}^\prime)-q_{21}(q_{11}-q_{22})^2
+2q_{12}q_{21}^2\right)
\,,
\end{split}
\end{equation*}
and $\frac{dq_{11}}{dt_{n}}=\frac{dq_{22}}{dt_{n}}=0$, for $n=1,2$.

Note that the Lenard-Magri recurrence condition \eqref{eq:LM-gl2}
holds for $u\in\mc V_1$ but it fails in general for $u\in\mc V$,
cf. Remark \ref{rem:V1-bi}.
For example, for $n=0$ the LHS of \eqref{eq:LM-gl2} is zero for every $u$,
while the RHS of \eqref{eq:LM-gl2} is not: for $u=q_{21}$ it is equal to $q_{21}$.
On the other hand, one can check that 
that this hierarchy coincides with the homogeneous Drinfeld-Sokolov hierarchy 
for $\mf{gl}_2$ associated to $S=\diag(s,0)$ and the diagonal matrix $a=\diag(1,0)$,
cf. e.g. \cite{DSKV13}.
In particular, this failure of the Lenard Magri scheme on $\mc V$ only occurs for $n=0$.

\subsection{Second case: \texorpdfstring{$S$}{S} nilpotent}
\label{sec:7.1b}

Consider the case when $S^t=\left(\begin{array}{ll}0&-1\\0&0\end{array}\right)$.
As before, 
the family of Adler type operators $A(\partial)+\epsilon S^t$, $\epsilon\in\mb F\backslash\{0\}$,
defines the compatible PVA $\lambda$-brackets 
$\{\cdot\,_\lambda\,\cdot\}_0$, given by \eqref{20150625:eq1},
and $\{\cdot\,_\lambda\,\cdot\}_1$ on $\mc V$,
whose only non-zero $\lambda$-brackets among generators are
\begin{equation}\label{20150625:eq1c}
\{{q_{11}}_{\lambda}q_{21}\}_1
=
\{{q_{12}}_{\lambda}q_{22}\}_1
=1
\,\,,\,\,\,\,
\{{q_{21}}_{\lambda}q_{11}\}_1
=
\{{q_{22}}_{\lambda}q_{21}\}_1
=-1
\,.
\end{equation}
Let us consider the canonical factorization $S^t=IJ$ where
$I=\begin{pmatrix}
1\\0
\end{pmatrix}$
and
$J=\begin{pmatrix}
0&-1
\end{pmatrix}$, and 
consider the corresponding $(IJ)$-quasideterminant of $A(\partial)$:
\begin{equation}\label{eq:Lnilp}
\begin{split}
& L(\partial)=|A(\partial)|_{IJ}
=(\partial+q_{11})q_{12}^{-1}(\partial+q_{22})-q_{21}
\\
& =\frac1{q_{12}}\partial^2
+\Big(\frac{q_{11}+q_{22}}{q_{12}}-\frac{q_{12}^\prime}{q_{12}^2}\Big)\partial
+\Big(\frac{q_{22}}{q_{12}}\Big)^\prime+\frac{q_{11}q_{22}}{q_{12}}-q_{21}
\,.
\end{split}
\end{equation}
This is a scalar differential operator with coefficients in a differential algebra $\mc V$
containing the inverse of $q_{12}$.

Since $L(\partial)$ is a differential operator,
the residue of any its power is zero.
On the oher hand, assuming that the differential algebra $\mc V$ contains the square root
of the element $q_{12}$,
we can consider the square root $B(\partial)\in\mc V((\partial^{-1}))$ of $L(\partial)$,
and use the method described in Section \ref{sec:6.3} to construct 
a bi-Hamiltonian integrable hierarchy for the two Poisson structures
associated to the PVA $\lambda$-brackets \eqref{20150625:eq1} and \eqref{20150625:eq1c}.

The first few coefficients of the operator
$B(\partial)=b_1\partial+b_0+b_{-1}\partial^{-1}+\dots$ are
\begin{align*}
&
b_1=\frac{1}{\sqrt{q_{12}}}\,,
\qquad
b_0=\frac{1}{2\sqrt{q_{12}}}(q_{11}+q_{22})-\frac{q_{12}'}{4\sqrt{q_{12}^3}}\,,
\\
&
b_{-1}=-\frac{1}{2\sqrt{1}}q_{21}\sqrt{q_{12}}
-\frac{1}{8\sqrt{q_{12}}}\left((q_{11}-q_{22})^2+2(q_{11}'-q_{22}')\right)\\
&
+\frac{1}{8\sqrt{q_{12}^3}}\left(2(q_{11}-q_{22})q_{12}'+q_{12}''\right)
-\frac{7}{32\sqrt{q_{12}^5}}(q_{12}')^2
\,.
\end{align*}
In particular, the first conserved density is
$h_{1,B}=-2\res_zB(z)=-2 b_{-1}$
and the first non-trivial equation of the hierarchy,
associated to the corresponding Hamiltonian functional, is
\begin{equation*}
\begin{split}
& 
\frac{dq_{12}}{dt_{1,B}}=\frac{d(q_{11}+q_{22})}{dt_{1,B}}=0
\\
& 
\frac{d(q_{11}-q_{22})}{dt_{1,B}}
=q_{21}\sqrt{q_{12}}
+\frac{1}{4\sqrt{q_{12}}}\left((q_{11}-q_{22})^2+2(q_{11}'-q_{22}')\right)
\\
&\qquad
-\frac{1}{4\sqrt{q_{12}^3}}\left(q_{12}''+2(q_{11}-q_{22})q_{12}'\right)
+\frac{7}{16\sqrt{q_{12}^5}}(q_{12}')^2\,,
\\
&
\frac{dq_{21}}{dt_{1,B}}
=\frac{1}{2\sqrt{q_{12}}}\left(q_{21}'-q_{21}(q_{11}-q_{22})\right)
\\
&\qquad
+\frac{1}{8\sqrt{q_{12}^3}}\left(2q_{12}'q_{21}-(q_{11}-q_{22})^3+2(q_{11}''-q_{22}'')-q_{12}'''\right)
\\
&\qquad
+\frac{1}{16\sqrt{q_{12}^5}}\left(3(q_{11}-q_{22})^2q_{12}'-6(q_{11}'-q_{22}')q_{12}'-2(q_{11}-q_{22})q_{12}''\right)
\\
&\qquad
+\frac{5}{32\sqrt{q_{12}^7}}\left((q_{11}-q_{22})(q_{12}')^2+4q_{12}'q_{12}''\right)
-\frac{35}{64\sqrt{q_{12}^9}}(q_{12}')^3
\,.
\end{split}
\end{equation*}

Note that, if we Dirac reduce (cf. \cite{DSKV14}) by the constraint $q_{12}=1$ and $q_{11}+q_{22}=0$
and set $w=-(q_{21}+\frac12 (q_{11}-q_{22})'+\frac14 (q_{11}-q_{22})^2)$,
the above equation gives
$\frac{dw}{dt_{1,B}}=w'$.
We omit the, rather cumbersome, computation of the next equation of the hierarchy,
but it is possible to check that,
under the same Dirac reduction and substitution,
we get the KdV equation 
$\frac{dw}{dt_{3,B}}=\frac{1}{3}w'''+2ww'$.
This is not a coincidence: as we will show in \cite{DSKVfuture},
the Drinfeld-Sokolov hierarchy for the affine $\mc W$-algebra
$\mc W(\mf{gl}_N,f)$, associated to a nilpotent element $f\in\mf{gl}_N$,
can be obtained, via a limiting procedure, by Dirac reduction of the 
bi-Hamiltonian hierarchy for the affine bi-PVA $\mc V_f(\mf{gl}_N)$.

\section{Dispersionless Adler type Laurent series and 
the corresponding dispersionless integrable hierarchies}
\label{sec:8}

\subsection{Quasi-classical limit of the algebra of pseudodifferential operators}
\label{sec:8.1}

Consider the algebra $\mc V((\partial^{-1}))$ of pseudodifferential operators
over the differential algebra $\mc V$.
Recall that the product $\circ$ on $\mc V((\partial^{-1}))$ is given,
in terms of the corresponding symbols, by \eqref{eq:prod-symbol}.
Let us consider the family of associative algebras $\mc V_\hbar((z^{-1}))$, $\hbar\in\mb F$,
whose underlying vector space is $\mc V((z^{-1}))$,
and the associative product is defined by the following deformation of \eqref{eq:prod-symbol}:
$$
(A\circ_{\hbar} B)(z)=A(z+\hbar\partial)B(z)\,,
\qquad A(z),B(z)\in\mc V((z^{-1}))
\,.
$$
We denote by $[\cdot\,,\,\cdot]_\hbar$ the commutator of this associative product.
Recall that (cf. \cite{LM79}),
the \emph{quasi-classical} limit $\mc V^{\text{qc}}((z^{-1}))$
of the family of associative algebras $\{\mc V_{\hbar}((z^{-1}))\}_{\hbar\in\mb F}$ is
the Poisson algebra defined as  the commutative associative algebra $\mc V((z^{-1}))$,
with the usual product of Laurent series,
endowed with the Poisson bracket
\begin{equation}\label{eq:quasi_symbol}
\{A(z),B(z)\}_{\text{qc}}=
\frac{d}{d\hbar}[A(z),B(z)]_\hbar\,\Big|_{\hbar=0}
=(\partial_z A(z))\partial B(z)-(\partial_z B(z))\partial A(z)
\,.
\end{equation}

\subsection{Dispersionless limit of a family of \texorpdfstring{$\lambda$}{lambda}-brackets and of an Adler type operator}
\label{sec:8.2}

Suppose that $A(\partial)\in\Mat_{M\times N}\mc V((\partial^{-1}))$ is a matrix pseudodifferential operator 
of $\hbar$-\emph{Adler type},
meaning that, for each $\hbar\in\mb F$ there exists a $\lambda$-bracket 
$\{\cdot\,_\lambda\,\cdot\}_{\hbar}$ on $\mc V$ such that
\begin{equation}\label{eq:eps-adler}
\begin{split}
\{A_{ij}(z)_\lambda A_{hk}(w)\}_\hbar
& =
A_{hj}(w+\hbar\lambda+\hbar\partial)
\iota_z(z\!-\!w\!-\!\hbar\lambda\!-\!\hbar\partial)^{-1}\big|_{x=\partial}A_{ik}(z-\hbar\lambda-\hbar x)
\\
& - A_{hj}(z)\iota_z(z-w-\hbar\lambda-\hbar\partial)^{-1}A_{ik}(w)
\,.
\end{split}
\end{equation}
We will be mainly interested in the scalar case $M=N=1$,
when the $\hbar$-Adler identity \eqref{eq:eps-adler} 
for the Laurent series $A(z)\in\mc V((z^{-1}))$ reads
\begin{equation}\label{eq:eps-scalar}
\begin{split}
\{A(z)_\lambda A(w)\}_\hbar
& = A(w+\hbar\lambda+\hbar\partial)
\iota_z(z-w-\hbar\lambda-\hbar\partial)^{-1}\big|_{x=\partial}A(z-\hbar\lambda-\hbar x)
\\
& - A(z)\iota_z(z-w-\hbar\lambda-\hbar\partial)^{-1}A(w)
\,.
\end{split}
\end{equation}
Note that \eqref{eq:eps-adler} with $\hbar=1$
says that $A(\partial)$ is of Adler type with respect to the $\lambda$-bracket 
$\{\cdot\,_\lambda\,\cdot\}_{\hbar=1}$.
Note also that the RHS of \eqref{eq:eps-scalar} is $0$ for $\hbar=0$.
We define the \emph{dispersionless limit}
of the family of $\lambda$-brackets $\{\cdot\,_\lambda\,\cdot\}_{\hbar}$ as
\begin{equation}\label{eq:disp-lambda}
\{\cdot\,_\lambda\,\cdot\}_{\text{disp}}
=
\frac{d}{d\hbar}\{\cdot\,_\lambda\,\cdot\}_{\hbar}\big|_{\hbar=0}
\,.
\end{equation}
Then, taking the derivative of both sides of \eqref{eq:eps-adler} with respect to $\hbar$
and letting $\hbar=0$,
we get that $A(z)$ satisfies the following \emph{dispersionless Adler} identity
with respect to the $\lambda$-bracket $\{\cdot\,_\lambda\,\cdot\}_{\text{disp}}$:
\begin{equation}\label{eq:disp-adler}
\begin{split}
& \{A_{ij}(z)_\lambda A_{hk}(w)\}_{\text{disp}}
= 
\iota_z(z-w)^{-1}
(\partial_w-\partial_z)\big(
A_{hj}(w) (\lambda+\partial) A_{ik}(z)
\big) \\
& +
\iota_z(z-w)^{-2}
\big(
A_{hj}(w)(\lambda+\partial)A_{ik}(z)
-
A_{hj}(z)(\lambda+\partial)A_{ik}(w)
\big)
\,.
\end{split}
\end{equation}
In the scalar case $M=N=1$, the dispersionless Adler identity 
for a Laurent series $A(z)\in\mc V((z^{-1}))$ reads
\begin{equation}\label{eq:quasi-adler}
\begin{split}
& \{A(z)_\lambda A(w)\}_{\text{disp}}
= 
\iota_z(z-w)^{-1}
(\partial_w-\partial_z)\big(
A(w) (\lambda+\partial) A(z)
\big) \\
& +
\iota_z(z-w)^{-2}
\big(
A(w)\partial A(z)
-
A(z)\partial A(w)
\big)
\,.
\end{split}
\end{equation}

\subsection{Dispersionless Adler type Laurent series and PVA's}
\label{sec:8.3}

\begin{definition}\label{def:disp-adler}
A Laurent series $A(z)\in\Mat_{N\times N}\mc V((z^{-1}))$ is called of \emph{dispersionless Adler type}
with respect to a $\lambda$-bracket $\{\cdot\,_\lambda\,\cdot\}_{\text{disp}}$
on the differential algebra $\mc V$ if equation \eqref{eq:disp-adler} holds.
Moreover, we say that $A(z)$ is of \emph{dispersionless bi-Adler type}
with respect to a pencil of $\lambda$-brackets 
$\{\cdot\,_\lambda\,\cdot\}_\epsilon=\{\cdot\,_\lambda\,\cdot\}_0+\epsilon\{\cdot\,_\lambda\,\cdot\}_1$
on $\mc V$,
if $A(\partial)+\epsilon\id_N$ is of dispersionless Adler type
with respect to $\{\cdot\,_\lambda\,\cdot\}_\epsilon$ for every $\epsilon\in\mb F$,
or, equivalently, if $A(z)$ is of dispersionless Adler type
with respect to $\{\cdot\,_\lambda\,\cdot\}_0$, and one has
\begin{equation}\label{eq:disp-bi-adler}
\begin{split}
& \{A_{ij}(z)_\lambda A_{hk}(w)\}_1
= 
\iota_z(z-w)^{-1}
\big(
\delta_{ik}\partial_w A_{hj}(w)\lambda-\delta_{hj}(\lambda+\partial)\partial_zA_{ik}(z)
\big) \\
& +
\iota_z(z-w)^{-2}
\big(
\delta_{ik}(A_{hj}(w)-A_{hj}(z))\lambda-\delta_{hj}(\lambda+\partial)(A_{ik}(w)-A_{ik}(z))
\big)
\,.
\end{split}
\end{equation}
\end{definition}
\begin{remark}\label{rem:disp}
By the arguments in Section \ref{sec:8.2}, if $A(\partial)\in\Mat_{M\times N}\mc V((\partial^{-1}))$ 
is an operator of $\hbar$-Adler type
with respect to a family of $\lambda$-brackets on $\mc V$, in the sense of \eqref{eq:eps-adler},
then its symbol $A(z)\in\Mat_{M\times N}\mc V((z^{-1}))$ is of dispersionless Adler type with respect
to the dispersionless limit $\lambda$-bracket \eqref{eq:disp-lambda}.
In fact, for a scalar Laurent series $A(z)\in\mc V((z^{-1}))$ 
the dispersionless Adler identity \eqref{eq:quasi-adler}
with respect to a $\lambda$-bracket $\{\cdot\,_\lambda\,\cdot\}_{\text{disp}}$ on $\mc V$
is equivalent to the condition that the $\hbar$-Adler identity \eqref{eq:eps-scalar}
holds modulo $\hbar^2$,
with respect to the $\lambda$-bracket 
$\{\cdot\,_\lambda\,\cdot\}_{\hbar}=\hbar\{\cdot\,_\lambda\,\cdot\}_{\text{disp}}$ 
on $\mc V$, extended by $\mb F[[\hbar]]$-linearity 
to the algebra of formal power series $\mc V[[\hbar]]$.
(Here we are using the fact that, 
for a scalar operator, the RHS of \eqref{eq:eps-scalar} is zero for $\hbar=0$.)
\end{remark}

The relation between dispersionless Adler type (scalar) Laurent series and Poisson vertex algebras
is given by the following result, which is analogous to Theorem \ref{thm:main-adler}.
\begin{theorem}\label{thm:main_quasi}
Let $A(z)\in\mc V((z^{-1}))$ be a Laurent series of dispersionless Adler type,
with respect to a $\lambda$-bracket $\{\cdot\,_\lambda\,\cdot\}$
on $\mc V$, and let $\mc V_1\subset\mc V$ be the differential subalgebra generated 
by the coefficients of $A(z)$.
Then $\{\cdot\,_\lambda\,\cdot\}$ is a PVA $\lambda$-bracket on $\mc V_1$.

If, moreover, $A(z)\in\mc V((z^{-1}))$ is of dispersionless bi-Adler type,
with respect to a pencil of $\lambda$-brackets 
$\{\cdot\,_\lambda\,\cdot\}_\epsilon=\{\cdot\,_\lambda\,\cdot\}_0+\epsilon\{\cdot\,_\lambda\,\cdot\}_1$,
$\epsilon\in\mb F$, on $\mc V$, 
then $\{\cdot\,_\lambda\,\cdot\}_\epsilon$, $\epsilon\in\mb F$ is a 
pencil of PVA $\lambda$-brackets on $\mc V_1$.
\end{theorem}
\begin{proof}
Of course one can prove Theorem \ref{thm:main_quasi}
by a straightforward but long computation.
Instead, we shall use a different argument.
The same proof of Theorem \ref{thm:main-adler}
(cf. the proofs of \cite[Lem.2.2 \& Lem.2.5]{DSKV15a})
can be adapted to prove that, if $A(\partial)\in\mc V((\partial^{-1}))$
satisfies the $\hbar$-Adler identity \eqref{eq:eps-scalar}
with respect to a family of $\lambda$-brackets
$\{\cdot\,_\lambda\,\cdot\}_\hbar$, $\hbar\in\mb F$,
then the skewsymmetry and Jacobi identity axioms hold
for these $\lambda$-brackets for every $\hbar\in\mb F$.
Moreover, still the same proof
shows that, if the $\lambda$-bracket $\{\cdot\,_\lambda\,\cdot\}_\hbar$
has values in $\hbar\mc V[[\hbar]]$
and $A(\partial)\in\mc V((\partial^{-1}))$
satisfies the $\hbar$-Adler identity \eqref{eq:eps-scalar}
modulo $\hbar^2$,
then the skewsymmetry axiom holds modulo $\hbar^2$:
$$
\{A(z)_\lambda A(w)\}_\hbar
+\{A(w)_{-\lambda-\partial}A(z)\}_\hbar
\equiv 0\mod(\hbar^2)
\,,
$$
and the Jacobi identity holds modulo $\hbar^3$:
\begin{equation*}
\begin{split}
& \{A(z_1)_\lambda \{A(z_2)_\mu A(z_3)\}_\hbar\}_\hbar
-\{A(z_2)_\mu \{A(z_1)_\lambda A(z_3)\}_\hbar\}_\hbar \\
& -\{{\{A(z_1)_\lambda A(z_2)\}_{\hbar}}_{\lambda+\mu} A(z_3)\}_\hbar
\equiv 0\mod(\epsilon^3)
\,.
\end{split}
\end{equation*}
The claim follows from Remark \ref{rem:disp}.
The arguments for the second statement are the same.
\end{proof}
\begin{example}\label{ex:generic-adler}
Consider the symbol of the generic matrix pseudodifferential (resp. differential) operators 
from Example \ref{ex:affine-generic}:
\begin{equation}\label{eq:generic-symbol}
L_{MN}(z)
=
\sum_{j=-\infty}^M U_j z^j
\qquad\Big(\text{resp. }\,\,
L_{(MN)}(z)
=
\sum_{j=0}^M U_j z^j
\Big)
\,\,\text{ with }\,\,
U_M=\id_N
\,,
\end{equation}
where $U_j=\big(u_{j;ab}\big)_{a,b=1}^N\in\Mat_{N\times N}\mc V$ for every $j<M$,
and $\mc V$ is the algebra of differential polynomials
in the infinitely many variables $u_{j;ab}$, 
where $1\leq a,b\leq N$ and $-\infty<j< M$
(resp. in the finitely many variables
$u_{j;ab}$, where $1\leq a,b\leq N$ and $0\leq j< M$).
For the same reasons that $L_{MN}(\partial)$ (resp. $L_{(MN)}(\partial)$)
is of bi-Adler type,
it is immediate that,
for every $\epsilon\in\mb F$,
$L_{MN}(z)+\epsilon\id_N$ (resp. $L_{(MN)}(z)+\epsilon\id_N$)
is a Laurent series (resp. polynomial) of $\hbar$-Adler type
with respect to a pencil of $\lambda$-brackets
$\{\cdot\,_\lambda\,\cdot\}_{\hbar,\epsilon}
=
\{\cdot\,_\lambda\,\cdot\}_{\hbar,0}+\epsilon\{\cdot\,_\lambda\,\cdot\}_{\hbar,1}$,
$\epsilon\in\mb F$,
depending on the parameter $\hbar\in\mb F$, on $\mc V$.
These $\lambda$-brackets are (uniquely and well) defined by the $\hbar$-Adler identity:
\begin{equation}\label{eq:generic-adler}
\begin{split}
& \{L_{MN;ij}(z)_\lambda L_{MN;hk}(w)\}_{\hbar,\epsilon} 
\\
&
= (L_{MN;hj}(w\!+\!\hbar\lambda\!+\!\hbar\partial)\!+\!\epsilon\delta_{hj})
\iota_z(z\!-\!w\!-\!\hbar\lambda\!-\!\hbar\partial)^{-1}\big|_{x=\partial}\!
(L_{MN;ik}(z\!-\!\hbar\lambda\!-\!\hbar x)\!+\!\epsilon\delta_{ik})
\\
& - (L_{MN;hj}(z)+\epsilon\delta_{hj})\iota_z(z-w-\hbar\lambda-\hbar\partial)^{-1}
(L_{MN;ik}(w)+\epsilon\delta_{ik})
\,,
\end{split}
\end{equation}
(and the same equations for $L_{(MN)}$ in place of $L_{MN}$).
It follows by the arguments in Section \ref{sec:8.2} that, 
taking the dispersionless limit \eqref{eq:disp-lambda},
the Laurent series $L_{MN}(z)$ (resp. the polynomial $L_{(MN)}(z)$)
is of dispersionless bi-Adler type with respect to the pencil of $\lambda$-brackets
$\{\cdot\,_\lambda\,\cdot\}_{\text{disp},\epsilon}
=\frac{d}{d\hbar}\{\cdot\,_\lambda\,\cdot\}_{\hbar,\epsilon}\big|_{\hbar=0}$.
By Theorem \ref{thm:main_quasi}
we have that, in the scalar case $N=1$,
this is a pencil of PVA $\lambda$-brackets on $\mc V$.
The two PVA structures are given, respectively, by
(denoting $L(z)$ either $L_{M1}(z)$ or $L_{(M1)}(z)$)
\begin{equation}\label{eq:generic-adler-0}
\begin{split}
& \{L(z)_\lambda L(w)\}_{\text{disp},0} 
= 
\iota_z(z-w)^{-2}
\big(
L(w)\partial L(z) 
- L(z)\partial L(w)
\big)
\\
& +
\iota_z(z-w)^{-1}
(\partial_w-\partial_z) \big( L(w) (\lambda+\partial) L(z) \big)
\,,
\end{split}
\end{equation}
and
\begin{equation}\label{eq:generic-adler-1}
\begin{split}
& 
\{L(z)_\lambda L(w)\}_{\text{disp},1} 
= 
\iota_z(z-w)^{-2} \partial (L(z)-L(w))
\\
& +\iota_z(z-w)^{-1} \big( 
\partial_w L(w)\lambda 
-(\lambda+\partial)\partial_zL(z)
\big)
\,.
\end{split}
\end{equation}
We can rewrite the two $\lambda$-brackets \eqref{eq:generic-adler-0}
and \eqref{eq:generic-adler-1} in terms of the generators $\{u_j\}_{j\leq M-1}$.
In the Laurent series case we have $L(z)=L_{M1}(z)=\sum_{j=-\infty}^Mu_jz^j$ ($u_M=1$),
and the two brackets \eqref{eq:generic-adler-0} and \eqref{eq:generic-adler-1}
are ($-\infty<i,j\leq M-1$)
\begin{equation}\label{eq:ga1}
\begin{split}
& \{{u_i}_\lambda u_j\}_{\text{disp},0} 
= 
(j+1)u_{j+1}(\lambda+\partial)u_{i+1}
\\
& +(j-i)\sum_{\ell=0}^{M-i-2}u_{j-\ell}(\lambda+\partial)u_{i+\ell+2}
-(2\lambda+\partial)\sum_{\ell=0}^{M-i-2}(\ell+1)u_{j-\ell}u_{i+\ell+2}
\,,
\end{split}
\end{equation}
and 
\begin{equation}\label{eq:ga2}
\{{u_i}_\lambda u_j\}_{\text{disp},1} 
= 
\left\{\begin{array}{ll}
-((i\!+\!1)(\lambda\!+\!\partial)\!+\!(j\!+\!1)\lambda)u_{i+j+2}
&\text{for } i,j\geq0\text{ s.t. } i\!+\!j\!+\!2\leq M \\
((i\!+\!1)(\lambda\!+\!\partial)\!+\!(j\!+\!1)\lambda)u_{i+j+2}
&\text{for }\, i,j\leq-1 \\
0 &\text{otherwise.}
\end{array}\right.
\end{equation}
The Poisson structure \eqref{eq:ga2} for $M=1$ was found in \cite{KM78}.
In the polynomial case we have $L(z)=L_{(M1)}(z)=\sum_{j=0}^Mu_jz^j$,
and the two brackets \eqref{eq:generic-adler-0} and \eqref{eq:generic-adler-1}
are ($0\leq i,j\leq M-1$)
\begin{equation}\label{eq:ga3}
\begin{split}
& \{{u_i}_\lambda u_j\}_{\text{disp},0} 
= 
(j+1)u_{j+1}(\lambda+\partial)u_{i+1}
\\
& +(j-i)\sum_{\ell=0}^{\min\{M-i-2,j\}}u_{j-\ell}(\lambda+\partial)u_{i+\ell+2}
-(2\lambda+\partial)\!\!\!\!\!\!
\sum_{\ell=0}^{\min\{M-i-2,j\}}\!\!\!\!\!\!
(\ell+1)u_{j-\ell}u_{i+\ell+2}
\,,
\end{split}
\end{equation}
and
\begin{equation}\label{eq:ga4}
\{{u_i}_\lambda u_j\}_{\text{disp},1} 
= 
\left\{\begin{array}{ll}
-((i\!+\!1)(\lambda\!+\!\partial)\!+\!(j\!+\!1)\lambda)u_{i+j+2}
&\text{if } i\!+\!j\!+\!2\leq M \\
0 &\text{otherwise}
\end{array}\right.
\end{equation}
Note that this bi-PVA is a quotient of the previous one by the differential ideal generated by $\{u_j\,|\,j<0\}$.

In the matrix case $N\geq2$ it is not true, in general, 
that $\{\cdot\,_\lambda\,\cdot\}_{\text{disp},\epsilon}$
is a pencil of PVA structures on $\mc V$,
since the limit $\hbar\to0$ of the RHS of \eqref{eq:generic-adler}
is, in general, not zero.
It is true, however, in the polynomial case of order $1$, i.e. for $L_{(1N)}(z)$.
In this case, the $0$-th $\lambda$-bracket $\{\cdot\,_\lambda\,\cdot\}_{\text{disp},0}$
is a PVA $\lambda$-bracket, which is given explicitly,
letting $L_{(1N)}(z)=\delta_{ij}z+u_{ij}$, by
\begin{equation}\label{eq:generic-adler-0-matrix}
\{{u_{ij}}_\lambda u_{hk}\}_{\text{disp},0} 
= 
\delta_{hj}\delta_{ik}
\lambda
\,,
\end{equation}
while the $1$-st $\lambda$-bracket $\{\cdot\,_\lambda\,\cdot\}_{\text{disp},1}$
is identically zero.
\end{example}

\subsection{Integrable hierarchies associated to dispersionless Adler type Laurent series}
\label{sec:8.4}

In this section we show that,
given a Laurent series of dispersionless Adler type with respect to a PVA $\lambda$-bracket,
there is an associated sequence of Hamiltonian functionals in involution,
namely the dipersionless analogue of Theorem \ref{thm:hn},
and that given a Laurent series of dispersionless bi-Adler type with respect 
to a pencil of PVA $\lambda$-brackets,
the corresponding sequence of Hamiltonian functionals satisfies
the Lenard-Magri recurrence relation with respect to the given Poisson structures,
namely the dipersionless analogue of Theorem \ref{thm:bi-hn}.
\begin{theorem}\label{thm:hn_quasi}
Let $\mc V$ be a differential algebra with a PVA
$\lambda$-bracket $\{\cdot\,_\lambda\,\cdot\}$.
Let $A(z)\in\mc V((z^{-1}))$
be a Laurent series of degree $K\geq1$ with invertible leading coefficient,
and let $B(z)\in\mc V((z^{-1}))$ be
its $K$-th root, i.e. $A(z)=B(z)^K$.
Assume that $A(z)$ is of dispersionless Adler type 
with respect to the $\lambda$-bracket $\{\cdot\,_\lambda\,\cdot\}$.
Define the elements $h_{n}\in\mc V$, $n\in\mb Z$, by
\begin{equation}\label{eq:hn_quasi}
h_{n}=
\frac{-K}{|n|}
\Res_z B(z)^n
\text{ for } n\neq0
\,,\,\,
h_0=0\,.
\end{equation}
Then: 
\begin{enumerate}[(a)]
\item
All the elements $\tint h_{n}$ are Hamiltonian functionals in involution:
\begin{equation}\label{eq:invol_quasi}
\{\tint h_{m},\tint h_{n}\}=0
\,\text{ for all } m,n\in\mb Z
\,.
\end{equation}
\item
The corresponding compatible hierarchy of Hamiltonian equations is
\begin{equation}\label{eq:hierarchy_quasi}
\frac{dA(z)}{dt_{n}}
=
\{\tint h_{n},A(z)\}
=
\{(B(z)^n)_+,A(z)\}_{\text{qc}}
\,,\,\,n\in\mb Z
\end{equation}
(the bracket $\{\cdot\,,\,\cdot\}_{\text{qc}}$ in the RHS is given by \eqref{eq:quasi_symbol}),
and the Hamiltonian functionals $\tint h_{n}$, $n\in\mb Z_+$, 
are integrals of motion of all these equations.
\item
If, moreover, $A(z)$ is a Laurent series of dispersionless bi-Adler type
with respect to two PVA $\lambda$-brackets 
$\{\cdot\,_\lambda\,\cdot\}_0$ and $\{\cdot\,_\lambda\,\cdot\}_1$ on $\mc V$,
then, the elements $h_{n}$, $n\in\mb Z_+$,
satisfy the Lenard-Magri recurrence relation:
\begin{equation}\label{eq:LM-K_quasi}
\{\tint h_{n},A(z)\}_0
=
\{\tint h_{n+K},A(z)\}_1
=
\{(B(z)^n)_+,A(z)\}_{\text{qc}}
\,\,,\,\,\,\,
n\in\mb Z
\,.
\end{equation}
Hence, \eqref{eq:hierarchy_quasi} is a compatible hierarchy of bi-Hamiltonian equations.
\end{enumerate}
\end{theorem}
The proof of Theorem \ref{thm:hn_quasi} relies on the next Lemmas \ref{lem:hn2_quasi} and
\ref{lem:hn3_quasi}.
Lemma \ref{lem:hn2_quasi} is analogous to Lemma \ref{lem:hn2}
and it is an obvious consequence of the Leibniz rules
(recall that the algebra $\mc V((z^{-1}))$ is commutative).
Lemma \ref{lem:hn3_quasi} is straightforward.
We state them for completeness.
\begin{lemma}\label{lem:hn2_quasi}
For $a\in\mc V$, $n\in\mb Z$, and $A(z)$ and $B(z)\in\mc V((z^{-1}))$ 
as in Theorem \ref{thm:hn_quasi}, we have
\begin{equation}\label{eq:hn2_quasi}
\begin{split}
& \{{h_{n}}_\lambda a\}\big|_{\lambda=0}
=
-\Res_z \{A(z)_x a\}\big(\big|_{x=\partial}B(z)^{n-K}\big)
\,,
\\
& \tint \{a_\lambda h_{n}\}\big|_{\lambda=0}
=
-\int \Res_z B(z)^{n-K}\{a_\lambda A(z)\}\big|_{\lambda=0}
\,.
\end{split}
\end{equation}
\end{lemma}
\begin{lemma}\label{lem:hn3_quasi}
For $A(z),B(z)\in\mc V((z^{-1}))$, we have
\begin{equation}\label{eq:residui_quasi}
\tint\res_z\{A(z),B(z)\}_{\text{qc}}=0
\,.
\end{equation}
\end{lemma}
\begin{proof}[Proof of Theorem \ref{thm:hn_quasi}]
Applying the second equation in \eqref{eq:hn2_quasi} first,
and then the first equation in \eqref{eq:hn2_quasi}, we get
\begin{equation}\label{eq:hn-pr1_quasi}
\{\tint h_{m},\tint h_{n}\}
= 
\int \Res_z \Res_w 
B(w)^{n-K}
\{A(z)_x A(w)\}
\big(\big|_{x=\partial}B(z)^{m-K}\big)
\,.
\end{equation}
We can now use the dispersionless Adler identity \eqref{eq:quasi-adler}
and the fact that $A(z)=B(z)^K$ to rewrite the RHS of \eqref{eq:hn-pr1_quasi} as
\begin{equation}\label{eq:hn-pr2_quasi}
\begin{split}
&
\int \Res_z \Res_w
\iota_z(z-w)^{-2}
\left(B(w)^nB(z)^{m-K}\partial A(z)
-B(z)^mB(w)^{n-K}\partial A(w)
\right)
\\
&
+\int \Res_z \Res_w
\left(
\partial(B(z)^m)B(w)^{n-K}\partial_wA(w)
-B(w)^{n}\partial(B(z)^{m-K}\partial_zA(z))
\right)
\,.
\end{split}
\end{equation}
Note that
\begin{equation}\label{20150721:eq4}
B(z)^{m-K}\partial A(z)=\frac{K}{m}\partial B(z)^m\,,
\end{equation}
and similarly with $\partial$ replaced by $\partial_z$ or $\partial_w$.
Hence, we can rewrite \eqref{eq:hn-pr2_quasi} as 
\begin{equation}\label{eq:hn-pr3_quasi}
\begin{split}
&
\int \Res_z \Res_w
\iota_z(z-w)^{-2}
\left(\frac Km B(w)^n\partial B(z)^{m}
-\frac Kn B(z)^m\partial B(w)^{n}
\right)
\\
&
+\int \Res_z \Res_w
\iota_z(z-w)^{-1}\left(
\frac Kn \partial(B(z)^m)\partial_wB(w)^{n}
-\frac Km B(w)^{n}\partial(\partial_zB(z)^{m})
\right)
\,.
\end{split}
\end{equation}
By the fact that $\partial$ and $\partial_z$ commute and using integration by parts,
we get
\begin{equation}\label{20150721:eq1}
\Res_z\iota_z(z-w)^{-1}\partial(\partial_zB(z)^{m})
=\Res_z \iota_z(z-w)^{-2}\partial B(z)^{m}
\,.
\end{equation}
Furthermore, by \eqref{eq:positive}, we have
\begin{equation}\label{20150721:eq2}
\res_z\iota_z(z-w)^{-1}\partial B(z)^m=\left(\partial B(w)^m\right)_+
=\partial(B(w)^m)_+
\,,
\end{equation}
where in the last identity we used the fact that $\partial$ acts componentwise 
on elements of $\mc V((w^{-1}))$, and
\begin{equation}\label{20150721:eq3}
\begin{split}
\res_z\iota_z(z-w)^{-2}B(z)^m
&=-\res_z\left(\partial_z\iota_z(z-w)^{-1}\right)B(z)^m
\\
&=\left(\partial_w B(w)^m\right)_+=\partial_w(B(w)^m)_+
\,.
\end{split}
\end{equation}
Here, in the second identity we used integration by parts
and equation \eqref{eq:positive}, and in the last identity the
fact that taking the positive part of elements in $\mc V((w^{-1}))$ commutes 
with the partial derivative $\partial_w$.
Hence, using equations \eqref{20150721:eq1}, \eqref{20150721:eq2} and \eqref{20150721:eq3} we can rewrite
equation \eqref{eq:hn-pr3_quasi} as
$$
\frac Kn \int \res_w\left(\partial(B(w)^m)_+\partial_w B(w)^n
-\partial_w(B(w)^m)_+\partial B(w)^n
\right)
$$
which is zero by Lemma \ref{lem:hn3_quasi} thus proving part (a).

For part (b),
by the first equation in Lemma \ref{eq:hn2_quasi} we have
\begin{equation}\label{eq:proofb-1_quasi}
\begin{split}
& \{\tint h_{n},A(w)\}
=
\{{h_{n}}_\lambda A(w)\}\big|_{\lambda=0}
=- \Res_z \{A(z)_x A(w)\}\big(\big|_{x=\partial}B(z)^{n-K}\big)
\,.
\end{split}
\end{equation}
By the dispersionless Adler assumption \eqref{eq:quasi-adler}, the identity $A=B^K$
and the first equation in \eqref{20150721:eq4}, the RHS of \eqref{eq:proofb-1_quasi} becomes
\begin{equation}\label{eq:proofb-2_quasi}
\begin{split}
& -\Res_z \iota_z(z-w)^{-2}\left(
\frac Kn A(w)\partial B(z)^n- B(z)^n \partial A(w)\right)
\\
&
- \Res_z \iota_z(z-w)^{-1}\left(
\partial_w(A(w))\partial B(z)^n-\frac Kn A(w)\partial\left(\partial_z B(z)^n\right)
\right)
\\
&=\partial_w(B(w)^n)_+\partial A(w)-\partial_w(A(w))\partial (B(w)^n)_+
\,.
\end{split}
\end{equation}
In the second equality we used equations \eqref{20150721:eq1}, \eqref{20150721:eq2} and \eqref{20150721:eq3}.
This proves \eqref{eq:hierarchy_quasi} and completes the proof of part (b).

Finally, we prove part (c).
By the first equation in \eqref{eq:hn2_quasi} we have
\begin{equation}\label{eq:bi-proof2_quasi}
\{\tint{h_{n}},A(w)\}_1
=
\{{h_{n}}_\lambda A(w)\}_1\big|_{\lambda=0}
=
-\Res_z
\{A(z)_x A(w)\}_1\big(\big|_{x=\partial}B^{n-K}(z)\big)
\,.
\end{equation}
Applying the dispersionless bi-Adler identity \eqref{eq:disp-bi-adler}
and \eqref{20150721:eq4}
we can rewrite the RHS of \eqref{eq:bi-proof2_quasi} as
\begin{equation}\label{eq:bi-proof3_quasi}
\begin{split}
& -\Res_z\iota_z(z-w)^{-1} 
\left(\frac Kn\partial B(z)^n-B(z)^{n-k}\partial A(w)\right)
\\
& -\Res_z 
\left(\partial_w A(w)\partial B(z)^{n-K}-\frac Kn\partial\partial_zB(z)^n\right)
\\
&= \{(B(z)^{n-K})_+,A(w)\}_{\text{qc}}
\,.
\end{split}
\end{equation}
In the first equality we used again equations \eqref{20150721:eq1}, 
\eqref{20150721:eq2} and \eqref{20150721:eq3}.
Comparing \eqref{eq:hierarchy_quasi} and \eqref{eq:bi-proof3_quasi},
we get \eqref{eq:LM-K_quasi}.
\end{proof}
\begin{remark}\label{rem:dispersionless_hierarchy}
Note that
the Poisson structures associated to the generic Adler type series in Example \ref{ex:generic-adler}
coincide with the corresponding dispersionless limit in \cite{DZ01}.
Hence, the hierarchy
\eqref{eq:hierarchy_quasi} can be obtained as the dispersionless limit, as defined by Dubrovin and Zhang,
of the hierarchy \eqref{eq:hierarchy}.
\end{remark}

\subsection{Examples}
\begin{example}[dKP hierarchy]\label{exa:dkp}
Let $\mc V$ be the algebra of differential polynomials in the variables
$\{u_i\mid i\leq0\}$.
The Laurent series
$L(z)=z+\sum_{i\leq0}u_iz^i$
is of dispersionless bi-Adler type with bi-PVA structure defined on generators by equations
\eqref{eq:ga1} and \eqref{eq:ga2} for $M=1$, see Example \ref{ex:generic-adler}.
We can compute
the first few integrals of motion $\tint h_{k}$, $k\geq1$, directly from the definition \eqref{eq:hn_quasi}:
\begin{equation*}
\begin{split}
& \tint h_1=-\tint u_{-1}
\,\,,\qquad
\tint h_2=-\tint (u_{-2}+u_0u_{-1})
\,, \\
& \tint h_3=-\tint (u_{-3}+2u_0u_{-2}+u_{-1}^2+u_0^2u_{-1})
\,,\dots
\end{split}
\end{equation*}
To find the corresponding bi-Hamiltonian equations,
we use Theorem \ref{thm:hn_quasi}.
We have
$L(z)_+=z+u_{0}$,
$L^2(z)_+=z^2+2u_{0}z+2u_{-1}+u_{0}^2$,
$L^3(z)_+=z^3+3u_{0}z^2+3(u_{-1}+u_{0}^2)z+(3u_{-2}+6u_0u_{-1}+u_{0}^3)$.
Hence, 
\begin{align}\label{dKP}
\begin{split}
\frac{dL(z)}{dt_{1}}
&
=\partial L(z)-\partial_zL(z)u_{0}'
\,,
\\
\frac{dL(z)}{dt_{2}}
&
=(2z+u_{0})\partial L(z)-2\partial_zL(z)(u_{0}'z+u_{-1}'+u_{0}u_{0}')
\,,
\\
\frac{dL(z)}{dt_{3}}
&
=\left(3z^2+6u_{0}z+3(u_{-1}+u_{0}^2)\right)\partial L(z)
\\
&
-\partial_z L(z)\left(3u_{0}'z^2+3(u_{-1}'+2u_{0}u_0')z+\partial(3u_{-2}+6u_{0}u_{-1}+u_{0}^3)\right)
\,,\dots
\end{split}
\end{align}
\end{example}
\begin{remark}\label{rem:kp}
Since $u_0$ is central for the $\lambda$-bracket \eqref{eq:ga2}, 
we have $\frac{du_{0}}{dt_{n}}=0$, for $n\in\mb Z_+$.
Hence, we may
assume $u_0=0$ in \eqref{dKP}.
The second equation in \eqref{dKP} then becomes
\begin{equation}\label{eq:benney}
\frac{d u_k}{dt_{2}}=2(u_{k-1}'-(k+1)u_{k+1}u_{-1}')
\,,
\qquad
k\leq-1
\,,
\end{equation}
which is the Benney equation for moments, see \cite{KM78,LM79}.
Moreover, consider the first two equations in the second system of the hierarchy \eqref{dKP},
and the first equation in the third system of \eqref{dKP}.
After eliminating the variables $u_{-2}$ and $u_{-3}$ 
and relabeling $t_2=y$, $t_3=t$ and $u=2u_{-1}$, we get
\begin{equation}\label{eq:kp}
3u_{yy}=(4u_t-6uu')'\,,
\end{equation}
which is known as the dispersionless Kadomtsev-Petviashvili equation.
\end{remark}

\begin{remark}
It has been shown by Radul \cite{Rad87} that there are infinitely 
many bi-Poisson structures for the KP hierarchy, see also \cite[Rem.3.14]{DSKV15a}.
In fact, we have infinitely many
bi-Poisson structures for the dKP hierarchy as well,
corresponding to the bi-PVA structures defined by \eqref{eq:ga1} and \eqref{eq:ga2} on the algebra 
of differential polynomials in the variables $\{u_i\mid i\leq M-1\}$, for every $M\geq1$.
These bi-PVA structures are obtained as the dispersionless limit of the bi-PVA structures 
for the KP equations.
\end{remark}

\begin{example}[The dispersionless $M$-th KdV hierarchy]
Let us consider the polynomial $L_{(M1)}(z)=\sum_{i=0}^Mu_iz^i$, with $u_{M}=1$.
It is of dispersionless bi-Adler type for the $\lambda$-brackets \eqref{eq:ga3} and \eqref{eq:ga4} 
defined on the algebra of differential
polynomials in the variables $\{u_i\mid 0\leq i\leq M-1\}$, see Example \ref{ex:generic-adler}.
It is proved in \cite{DSKV15a} that, after setting $u_{M-1}=0$ 
(which is possible since $u_{M-1}$ is central for the $\lambda$-bracket \eqref{eq:ga4}),
the integrable bi-Hamiltonian hierarchy \eqref{eq:LM-K}, 
associated to the Adler type differential operator
$L_{(M1)}(\partial)$, is the $M$-th KdV hierarchy.
Hence, the integrable bi-Hamiltonian hierarchy \eqref{eq:LM-K_quasi}
is the dispersionless $M$-th KdV hierarchy.
\end{example}

\section{Adler type pseudodifferential operators for double Poisson vertex algebras}
\label{sec:9}

In this section, by a differential algebra $\mc V$ we mean a (possibly noncommutative) 
unital associative algebra with a derivation $\partial$.
We refer to \cite{DSKV15b}
for the definition and the main properties of a structure of a double Poisson vertex algebra on $\mc V$.

Given a differential algebra $\mc V$, we define for each positive integer $N$
a unital commutative associative differential algebra $\mc V_N$, 
generated by elements $a_{ij}$, $a\in\mc V$, $i,j\in\{1,\dots,N\}$,
subject to the following relations ($k\in\mb F$, $a,b\in\mc V$, $i,j=1,\dots,N$)
\begin{equation}\label{questaserve}
(ka)_{ij}=ka_{ij}
\,\,,\,\,\,\,
(a+b)_{ij}=a_{ij}+b_{ij}
\,\,,\,\,\,\,
(ab)_{ij}=\sum_{k=1}^Na_{ik}b_{kj}\,,
\end{equation}
with derivation $\partial:\,\mc V_N\to\mc V_N$ given by $\partial a_{ij}=(\partial a)_{ij}$.
Recall from \cite[Prop.3.20, Thm.3.22]{DSKV15b} that, if $\mc V$ is a double Poisson vertex algebra, 
with $2$-fold $\lambda$-bracket $\ldb\cdot\,_\lambda\,\cdot\rdb$, written (in Sweedler's notation) as
\begin{equation}\label{sweedler}
\ldb a_\lambda b\rdb=\sum_{n\in\mb Z_+}(a_nb)'\otimes (a_nb)''\,\lambda^n
\,,
\end{equation}
then we have a PVA structure on $\mc V_N$, $N\geq1$,
with the $\lambda$-bracket given by
\begin{equation}\label{20130917:eq1}
\{a_{ij}{}_\lambda b_{hk}\}=\sum_{n\in\mb Z_+}(a_nb)'_{hj}(a_nb)''_{ik}\lambda^n\,.
\end{equation}

Recall also that, given a double PVA $\mc V$, we have a well defined Lie algebra structure on the quotient
space $\mc V/([\mc V,\mc V]+\partial\mc V)$ defined by
\begin{equation}\label{eq:double_lie}
\{\tint a,\tint b\}=\tint \mult\ldb a_{\lambda}b\rdb\big|_{\lambda=0}
=
\tint (a_0b)'(a_0b)''
\,,
\qquad
a,b\in\mc V\,,
\end{equation}
where $\tint:\mc V\to\mc V/([\mc V,\mc V]+\partial\mc V)$ denotes the projection map
and $\mult:\mc V\otimes\mc V\to\mc V$ denotes the multiplication map.
Furthermore, we have a representation of the Lie algebra 
$\mc V/([\mc V,\mc V]+\partial\mc V)$ on $\mc V$, with the Lie algebra action given by
\begin{equation}\label{eq:double_lie_action}
\{\tint a,b\}=\mult\ldb a_{\lambda}b\rdb\big|_{\lambda=0}
=
(a_0b)'(a_0b)''
\,,
\qquad
a,b\in\mc V\,.
\end{equation}

We shall also denote (for clarity of the exposition), for every $N\geq1$,
$\tint_N:\mc V_N\to\mc V_N/\partial\mc V_N$ the canonical quotient map and by $\{\cdot\,,\,\cdot\}^N$
the Lie bracket induced by \eqref{20130917:eq1} on the quotient space $\mc V_N/\partial\mc V_N$.
\begin{proposition}\phantomsection\label{20150722:prop1}
\begin{enumerate}[(a)]
\item
There is a well defined Lie algebra homomorphism
$\varphi_N:\,\mc V/([\mc V,\mc V]+\partial\mc V)\to\mc V_N/\partial\mc V_N$ given by
$\varphi_N(\tint a)=\int_N\sum_{i=1}^Na_{ii}$.
\item
For every $i,j=1,\dots,N$, the map $a\mapsto a_{ij}$,
from the module $\mc V$ over the Lie algebra $\mc V/([\mc V,\mc V]+\partial\mc V)$
to the module $\mc V_N$ over the Lie algebra $\mc V_N/\partial\mc V_N$,
is compatible with the Lie algebra homomorphism 
$\varphi_N:\,\mc V/([\mc V,\mc V]+\partial\mc V)\to\mc V_N/\partial\mc V_N$
defined in (a), meaning that
\begin{equation}\label{eq:varphiN-comp}
\{\tint a,b\}_{ij}=\{\varphi_N(\tint a),b_{ij}\}
\,\,,\,\,\,\,\text{ for all } a,b\in\mc V
\,.
\end{equation}
\end{enumerate}
\end{proposition}
\begin{proof}
Straightforward.
\end{proof}
\begin{remark}\label{LCA}
In fact, if $\mc V$ is a double Poisson vertex algebra,
then we have a Lie conformal algebra $\lambda$-bracket on $\mc V/[\mc V,\mc V]$,
given by $\{\tr(a)_\lambda\tr(b)\}=\tr\left(\mult\ldb a_\lambda b\rdb\right)$, $a,b\in\mc V$,
where $\tr:\mc V\to\mc V/[\mc V,\mc V]$ is the canonical quotient map,
and we have a Lie conformal algebra homomorphism 
$\mc V/[\mc V,\mc V]\to \mc V_N$,
given by $\tr(a)\mapsto\sum_{i=1}^Na_{ii}$.
\end{remark}
\begin{definition}[\cite{DSKV15b}]\label{20131217:def1}
Let $\mc V$ be a differential algebra endowed with a $2$-fold $\lambda$-bracket
$\ldb\cdot\,_\lambda\,\cdot\rdb$.
We call a pseudodifferential operator $A(\partial)\in\mc V((\partial^{-1}))$
of \emph{Adler type} for $\ldb\cdot\,_\lambda\,\cdot\rdb$ if the following identity holds in
$\mc V^{\otimes2}[\lambda,\mu]((z^{-1},w^{-1}))$:
\begin{equation}\label{20130923:eq2}
\begin{array}{c}
\displaystyle{
\ldb A(z)_{\lambda}A(w)\rdb
=A(w+\lambda+\partial)\otimes i_z(z-w-\lambda-\partial)^{-1}A^*(-z+\lambda)
}
\\
\displaystyle{
-A(z)\otimes i_z(z-w-\lambda-\partial)^{-1}A(w)
\,.
}
\end{array}
\end{equation}
Furthermore, we say that $A(\partial)$ is of bi-Adler type
with respect to the two $2$-fold $\lambda$-brackets $\ldb\cdot\,_\lambda\,\cdot\rdb_0$
and $\ldb\cdot\,_\lambda\,\cdot\rdb_1$
if it is of Adler type with respect to the $2$-fold $\lambda$-bracket $\ldb\cdot\,_\lambda\,\cdot\rdb_0$,
i.e. \eqref{20130923:eq2} holds,
and  
\begin{equation}\label{eq:bi-adler-dpva}
\begin{split}
\ldb A(z)_\lambda A(w)\rdb_1
& =1\otimes \iota_z(z-w-\lambda-\partial)^{-1}\left( A^*(-z+\lambda)- A(w)\right)
\\
& + \iota_z(z-w-\lambda)^{-1}\left( A(w+\lambda)-A(z)\right)\otimes1
\,.
\end{split}
\end{equation}
\end{definition}
The following results are a generalization of the results in \cite[Sec.5]{DSKV15b} and can be proved 
adapting the proofs of Theorems \ref{thm:main-adler}, \ref{thm:hn}, \ref{thm:main-bi-adler} 
and \ref{thm:bi-hn} to the double Poisson vertex algebra setting.
\begin{theorem}
Let $\mc V$ be a differential algebra, endowed with a $2$-fold $\lambda$-bracket $\ldb\cdot\,_\lambda\,\cdot\rdb$.
Let $A(\partial)\in\mc V((\partial^{-1}))$ be a pseudodifferential operator of Adler type with respect
to $\ldb\cdot\,_\lambda\,\cdot\rdb$.
Let $\mc V_1\subset\mc V$ be the differential subalgebra generated by the coefficients of $A(\partial)$.
Then:
\begin{enumerate}[(a)]
\item 
$\mc V_1$ is a double Poisson vertex algebra (with double $\lambda$-bracket $\ldb\cdot\,_\lambda\,\cdot\rdb$).
\item
For $B(\partial)\in\mc V((\partial^{-1}))$
a $K$-th root of $A$ (i.e. $A(\partial)=B(\partial)^K$ for $K\in\mb Z\backslash\{0\}$)
define the elements $h_{n,B}\in\mc V$, $n\in\mb Z$, by
\begin{equation}\label{eq:hn_double}
h_{n,B}=
\frac{-K}{|n|}
\Res_z B^n(z)
\text{ for } n\neq0
\,,\,\,
h_0=0\,.
\end{equation}
Then all the elements $\tint h_{n,B}$ are Hamiltonian functionals in involution:
$$
\{\tint h_{m,B},\tint h_{n,C}\}=0
\,\,,\,\,\,\, \text{ for all } m,n\in\mb Z\,,\,\, B,C \text{ roots of } A\,.
$$
The corresponding compatible hierarchy of Hamiltonian equations is
\begin{equation}\label{eq:hierarchy_dpva}
\frac{dA(z)}{dt_{n,B}}
=
\{\tint h_{n,B},A(z)\}
=
[(B^n)_+,A](z)
\,,\,\,n\in\mb Z,\,
B \text{ root of } A
\,,
\end{equation}
and the Hamiltonian functionals $\tint h_{n,C}$, $n\in\mb Z_+$, $C$ root of $A$,
are integrals of motion of all these equations.
\end{enumerate}
Furthermore, let us assume that $A(\partial)\in\mc V((\partial^{-1}))$
is of bi-Adler type 
with respect to the $2$-fold $\lambda$-brackets $\ldb\cdot\,_\lambda\,\cdot\rdb_0$ 
and $\ldb\cdot\,_\lambda\,\cdot\rdb_1$ on $\mc V$.
Then:
\begin{enumerate}[(a)]
\setcounter{enumi}{2}
\item
$\mc V$ is a bi-double PVA with the $2$-fold $\lambda$-brackets $\ldb\cdot\,_\lambda\,\cdot\rdb_0$
and $\ldb\cdot\,_\lambda\,\cdot\rdb_1$.
\item
For a $K$-th root $B(\partial)\in\mc V((\partial^{-1}))$ of $A(\partial)$,
the elements $h_{n,B}\in\mc V$, $n\in\mb Z_+$, given by \eqref{eq:hn_double}
satisfy the generalized Lenard-Magri recurrence relation \eqref{eq:LM-K}.
Hence, \eqref{eq:hierarchy_dpva} is a compatible hierarchy of bi-Hamiltonian equations.
\end{enumerate}
\end{theorem}
Note that by Proposition \ref{20150722:prop1}(a), the sequence $h_{n,B}\in\mc V$, $n\in\mb Z_+$,
defined by \eqref{eq:hn_double} gives rise to a sequence
$$
h_{n,B}^N=\varphi_N(\tint h_{n,B})\,\in\mc V_N\,,
\qquad
n\in\mb Z_+,\, N\geq1
\,,
$$
of Hamiltonian functionals in involution,
and by Proposition \ref{20150722:prop1}(b) we get the integrable Hamiltonian hierarchy
in $\mc V_N$ corresponding to \eqref{eq:hierarchy_dpva}, which can be regarded as the Hamiltonian hierarchy 
for the entries of the matrix hierarchy \eqref{eq:hierarchy_dpva}.
\begin{example}[{\cite[Ex.3.23]{DSKV15b}}]
Consider the noncommutative algebra of differential polynomials in two (noncommutative) variables $u$ and $v$
endowed with the following $2$-fold PVA $\lambda$-bracket:
$$
\ldb u_\lambda u\rdb=1\otimes u-u\otimes1+c(1\otimes1)\lambda\,,
\qquad c\in\mb F\,,\qquad
v \text{ central }
\,.
$$ 
The elements
$h_0=1,\, h_n=\frac1n (u+v)^n$ for $n>0$, are in involution, 
and the corresponding integrable hierarchy of Hamiltonian equations is:
$$
\frac{du}{dt_n} =v(u+v)^n u -u(u+v)^n v +c\partial (u+v)^{n+1}
\,\,,\,\,\,\,
\frac{dv}{dt_n}=0
\,,\,\, n\in \mb Z_+
\,.
$$
For every $N\geq1$, $\mc V_N$ is the commutative algebra of differential polynomials in the variables
$u_{ij}$ and $v_{ij}$, $1\leq i,j\leq N$, endowed with the following PVA $\lambda$-bracket
$$
\{{u_{ij}}_\lambda u_{hk}\}=\delta_{hj}u_{ik}-\delta_{ik}u_{hj}+c\delta_{ik}\delta_{hj}\lambda
\,,\qquad
v_{ij} \text{ central for every } i,j
\,.
$$
The corresponding Hamiltonian densities are
$$
h_0^N=N
\,\,,\,\,\,\,
h_n^N=\frac1n\sum_{i,j=1}^N\sum_{k=0}^n\binom{n}{k}(u^k)_{ij}(v^{n-k})_{ji}
\,,
\qquad
n\geq1
\,,
$$
where in the RHS above we have used equation \eqref{questaserve}.
The Hamiltonian equations become
\begin{align*}
\frac{du_{ij}}{dt_n}
&
=
c\partial ((u+v)^{n+1})_{ij}
+\sum_{h,k=1}^N\left(
v_{ih}((u+v)^n)_{hk} u_{kj} -u_{ih}((u+v)^n)_{hk}v_{kj} \right)
\,,
\\
\frac{dv_{ij}}{dt_n}
&=0
\,,
\qquad
1\leq i,j\leq N, n\in \mb Z_+
\,.
\end{align*}
\end{example}

\appendix

\section{Square roots of \texorpdfstring{$\partial\id_2+Q\in\Mat_{2\times2}\mc V((\partial^{-1}))$}{d+Q}}
\label{sec:app}

Let $\mc K$ be a differential field
with subfield of constant $\mb F$.
Let $Q=\big(q_{ji}\big)_{i,j=1}^2\in\Mat_{2\times2}\mc K$.
The matrix differential operator $A(\partial)=\partial\id_2+Q\in\Mat_{2\times2}\mc K[\partial]$
has Dieudonn\`e determinant of degree $2$,
and it is natural to ask whether $A(\partial)$ has a square root in 
$\Mat_{2\times2}\mc K((\partial^{-1}))$
or over an extension of $\mc K$.
\begin{lemma}\label{lem:roots1}
In a differential field extension $\widetilde{\mc K}$ of $\mc K$,
there exists an invertible matrix $B\in\Mat_{N\times N}\widetilde{\mc K}$
such that
\begin{equation}\label{eq:roots1}
B^\prime+QB=0
\,.
\end{equation}
\end{lemma}
\begin{proof}
Equation \eqref{eq:roots1} can be rewritten as the following systems
of $N^2$ $1$-st order linear differential equations in the $N^2$ entries $b_{ij}$ of the matrix $B$:
$$
b_{ij}^\prime+\sum_{\ell=1}^Nq_{\ell i}b_{\ell j}=0
\,,
$$
which has an $N^2$ dimensional space of solutions in a differential field extension
$\widetilde{\mc K}$ of $\mc K$.
\end{proof}
\begin{lemma}\label{lem:roots2}
All the possible square roots of the matrix $\id_2\partial$ 
are the matrix pseudodifferential operators with constant coefficients
$\Delta\in\Mat_{2\times 2}\mb F((\partial^{-1}))$
such that $\tr(\Delta)=0$ and $\det(\Delta)=\partial$.
Namely, they are the matrices of the form
\begin{equation}\label{eq:roots2}
\Delta=
\left(\begin{array}{ll}
\alpha(\partial) & \beta(\partial) \\
\beta(\partial)^{-1}(\alpha(\partial)^2-\partial) & -\alpha(\partial)
\end{array}\right)
\,.
\end{equation}
with $\alpha,\beta\in\mb F((\partial^{-1}))$ and $\beta(\partial)\neq0$.
\end{lemma}
\begin{proof}
Let $\Delta=\left(\begin{array}{ll}
a(\partial) & b(\partial) \\
c(\partial) & d(\partial)
\end{array}\right)$,
be a square root of $\partial\id_2$.
In other words, the pseudodifferential operators $a,b,c,d\in\mc K((\partial^{-1}))$
satisfy
\begin{equation}\label{eq:roots3}
\begin{split}
& a(\partial)^2+b(\partial)c(\partial)=\partial \\
& a(\partial)b(\partial)+b(\partial)d(\partial) = 0 \\
& c(\partial)a(\partial)+d(\partial)c(\partial) = 0 \\
& c(\partial)b(\partial)+d(\partial)^2=\partial
\,.
\end{split}
\end{equation}
It cannot be $b(\partial)=0$, 
otherwise from the first equation we get $a(\partial)^2=\partial$,
which is impossible by obvious degree considerations.
Hence, $b(\partial)\neq0$,
and from the second equation we get
\begin{equation}\label{eq:roots4}
d(\partial)=-b(\partial)^{-1}a(\partial)b(\partial)
\,.
\end{equation}
Substituting \eqref{eq:roots4}
in the other three equations of \eqref{eq:roots3}, we get
\begin{equation}\label{eq:roots5}
\begin{split}
& a(\partial)^2+b(\partial)c(\partial)=\partial \\
& c(\partial)a(\partial)-b(\partial)^{-1}a(\partial)b(\partial)c(\partial) = 0 \\
& c(\partial)b(\partial)+b(\partial)^{-1}a(\partial)^2b(\partial)=\partial
\,.
\end{split}
\end{equation}
If we multiply the third equation of \eqref{eq:roots5} on the left by $b(\partial)$ 
and on the right by $b(\partial)^{-1}$,
we get
$$
b(\partial)c(\partial)+a(\partial)^2=b(\partial)\partial b(\partial)^{-1}
\,,
$$
which, when compared to the first equation of \eqref{eq:roots5}, implies
\begin{equation}\label{eq:roots7}
b(\partial)\partial=\partial b(\partial)\,,
\end{equation}
namely $b(\partial)\in\mb F((\partial^{-1}))$.
Next, if we multiply on the left by $b(\partial)$ the second equation of \eqref{eq:roots5}
we get
$$
b(\partial)c(\partial)a(\partial)=a(\partial)b(\partial)c(\partial)\,,
$$
and this equation can be rewritten, using the first equation of \eqref{eq:roots5},
as
$$
\partial a(\partial)=a(\partial)\partial\,,
$$
namely $a(\partial)\in\mb F((\partial^{-1}))$.
To conclude, $a(\partial)$ and $b(\partial)$, being with constant coefficients, commute, 
and therefore equation \eqref{eq:roots4} gives $d(\partial)=-a(\partial)$.
Moreover, the first equation of \eqref{eq:roots5} gives
$c(\partial)=-b(\partial)^{-1}(\partial-a(\partial)^2)$.
This proves \eqref{eq:roots2} and the claim.
\end{proof}
\begin{remark}\label{rem:roots2}
It is not true, in general,
that a matrix pseudodifferential operator $A(\partial)\in\Mat_{N\times N}\mc V((\partial^{-1}))$
whose degree of the Dieudonne determinant is equal to $K\geq1$
admits a $K$-th root.
For example, the matrix 
$\left(\begin{array}{ll} \partial^3 & 0 \\ 0 & \partial \end{array}\right)$
has Dieudonne determinant of degree $4$.
But the same computations used in the proof of Lemma \eqref{lem:roots2}
show that this matrix cannot have any square root
(equation \eqref{eq:roots7}
is replaced, in this case, by equation $b(\partial)\partial^3=\partial b(\partial)$,
which does not admit any solution, by obvious degree considerations).
\end{remark}
\begin{proposition}\label{thm:roots}
Let $B\in\Mat_{2\times 2}\widetilde{\mc K}$
be any invertible matrix solving \eqref{eq:roots1}.
Then all possible square roots of the matrix
$A(\partial)=\partial+Q$ 
are of the form 
\begin{equation}\label{eq:roots6}
A^{\frac12}(\partial)
=
B\Delta(\partial)B^{-1}
\,,
\end{equation}
where $\Delta(\partial)$ is an arbitrary root of $\partial\id_2$,
as in Lemma \ref{lem:roots2}.
\end{proposition}
\begin{proof}
Note that, under the assumption \eqref{eq:roots1} on $B$,
we have $B\partial B^{-1}=\partial+Q$.
Hence, if $A^{\frac12}(\partial)$ is a square root of $A(\partial)$,
then 
$$
(A^{\frac12}(\partial))^2
=
\partial+Q
=
B\partial B^{-1}
\,,
$$
and multiplying on the left by $B^{-1}$ and on the right by $B$, we get
$$
(B^{-1}A^{\frac12}(\partial)B)^2
=
\partial\id_2
\,.
$$
In other words, $B^{-1}A^{\frac12}(\partial)B$ is a square root of $\partial\id$,
i.e. it is $B^{-1}A^{\frac12}B=\Delta(\partial)$,
as in Lemma \ref{lem:roots2}.
The claim follows.
\end{proof}
\begin{remark}\label{rem:roots1}
Consider the local density functions obtained by taking residue of fractional powers of the operator $A(\partial)$
(cf. \eqref{eq:hn}):
$$
h_n
=
\tr\Res_\partial
A^{\frac n2}(\partial)
\,,\,\,n\geq0\,,
$$
where $A^{\frac12}(\partial)=(\partial\id_2+Q)^{\frac12}$ is one of the square roots given by Proposition \ref{thm:roots}.
For even $n=2k$, $A^{\frac n2}(\partial)=B\partial^kB^{-1}$,
so the residue is zero unless $n=-2$.
On the other hand, for odd $n=2k+1$, we have
\begin{equation}\label{eq:remark}
h_n
=
\tr\Res_\partial
B\Delta^n(\partial)B^{-1}
=
\tr\Res_\partial
B\Delta(\partial)\partial^kB^{-1}
\,.
\end{equation}
By the cyclic property $\tint\tr\Res_\partial AB=\tint\tr\Res_\partial BA$,
and since, by Lemma \ref{lem:roots2}, $\Delta(\partial)$ has zero trace,
the element $h_n$ lies in the image of $\partial$ (in $\widetilde{\mc K}$).
%
\end{remark}


\end{document}